\documentclass[a4paper,onecolumn,11pt,accepted=2021-10-29]{quantumarticle}
\pdfoutput=1
\usepackage{authblk}

\usepackage{graphicx}
\usepackage{fullpage}
\usepackage[utf8]{inputenc}
\usepackage{hyperref,dcolumn,bbm,url,booktabs,braket,microtype,aliascnt,mathtools,mathrsfs,etoolbox,cancel,capt-of,algpseudocode,float,newfloat,amsmath,amssymb,color,amsthm,cleveref,tikz, dsfont,subcaption}
\usepackage[T1]{fontenc}
\usepackage[USenglish]{babel}
\usepackage[numbers,sort&compress]{natbib}
\usepackage{qcircuit}
\usepackage{url}
\usepackage{xcolor}

\newtheorem{theorem}{Theorem}[section]
\newtheorem{lemma}[theorem]{Lemma}

\newtheorem{definition}[theorem]{Definition}

\newcommand{\bb}{\mathbb}
\newcommand{\mf}{\mathfrak}
\newcommand{\mc}{\mathcal}

\def\cA{\mathcal{A}}

\def\cC{\mathcal{C}}

\def\cH{\mathcal{H}}
\def\cI{\mathcal{I}}

\def\cO{\mathcal{O}}
\def\cP{\mathcal{P}}
\def\cQ{\mathcal{Q}}

\def\cS{\mathcal{S}}

\def\one{{\mathchoice {\rm 1\mskip-4mu l} {\rm 1\mskip-4mu l} {\rm
			1\mskip-4.5mu l} {\rm 1\mskip-5mu l}}}

\newcommand{\ketbra}[1]{\, | #1 \rangle \hspace{-0.05em} \langle #1 | \,} 
\newcommand{\ketbras}[2]{\, | #1 \rangle \hspace{-0.05em} \langle #2 | \,}

\newcommand{\poly}{\operatorname{poly}}

\begin{document}
	
	\title{Fast-forwarding quantum evolution}
	
	

	\author[1]{
		Shouzhen Gu}

	\author[2]{
		Rolando D. Somma}

	\author[2]{Burak \c{S}ahino\u{g}lu}
	
	\affil[1]{Institute for Quantum Information and Matter, California Institute of Technology, Pasadena, CA 91125, USA}
	
	\affil[2]{Theoretical Division, Los Alamos National Laboratory, Los Alamos, NM 87545, USA}

	\maketitle

	\begin{abstract}
		We investigate the problem of fast-forwarding quantum evolution, whereby the dynamics of certain quantum systems can be simulated with gate complexity that is sublinear in the evolution time.  We provide a definition of fast-forwarding that considers the model of quantum computation, the Hamiltonians that induce the evolution, and the properties of the initial states.
		Our definition accounts for {\em any} asymptotic complexity improvement of the general case and we use it to demonstrate fast-forwarding in several quantum systems. In particular, we show that some local spin systems whose Hamiltonians can be taken into block diagonal form using an efficient quantum circuit, such as those
		that are permutation-invariant, can be exponentially fast-forwarded. 
		We also show that certain classes of positive semidefinite local spin systems, also known as frustration-free, can be polynomially fast-forwarded, provided the initial state is supported on a subspace of sufficiently low energies. 
		Last, we show that all quadratic fermionic systems and number-conserving quadratic bosonic systems can be exponentially fast-forwarded in a model where quantum
		gates are exponentials of specific fermionic or bosonic operators, respectively.
		Our results extend the classes of physical Hamiltonians that were previously known to be fast-forwarded, while 
		not necessarily requiring methods that diagonalize the Hamiltonians efficiently.
		We further develop a connection between fast-forwarding and precise energy measurements that also accounts for polynomial improvements.
	\end{abstract}

	\newpage	
	
	\tableofcontents

	\section{Introduction}
	\label{sec:intro}

	Since Feynman's initial proposal of simulating  physics with quantum computers~\cite{Fey82}, a significant amount of research in quantum computing has focused on quantum algorithms for Hamiltonian simulation, i.e., the problem of simulating the time evolution of a quantum system induced by a Hamiltonian $H$. 
	The main figure of efficiency is how the quantum complexity of these algorithms, as determined by the number of elementary (e.g., two-qubit) gates, scales with various parameters such as the evolution time $t>0$. 
	To date, the most efficient known methods for Hamiltonian simulation  have quantum complexities that are almost linear in $t$ (cf.,~\cite{Llo96,SOGKL02,BAC07,WBH+10,CW12,BC12,BCC+14,BCC+15,BCK15,LC17,LC19, campbell2019random,berry2020time}).
	These methods are thus optimal since simulating {\em all} quantum systems with complexity sublinear in $t$ is impossible~\cite{BAC07,AA17,HHK+18}. 
	
	Nevertheless, a natural question is whether such complexities can be significantly reduced in specific cases, opening up the possibility of \emph{fast-forwarding} Hamiltonian simulation.
	Roughly, a quantum system or Hamiltonian is said to be fast-forwarded if its time evolution can be simulated with quantum complexity scaling sublinearly in $t$.
	An answer to this question is relevant for both foundational and practical reasons.
	For example, the ability to fast-forward a Hamiltonian might allow us to study physical systems on quantum computers more rapidly. 
	It could also result in other important quantum speedups, as some well-known quantum algorithms use Hamiltonian simulation as a subroutine~(cf., \cite{HHL09,CKS17}).

	Indeed, a few examples of {\em exponential fast-forwarding} are known~\cite{AA17}, in which there exist quantum algorithms for Hamiltonian simulation of quantum complexity at most polylogarithmic in $t$.
	These examples include Hamiltonians that can be diagonalized via an efficient quantum circuit, commuting local Hamiltonians, and quadratic fermionic Hamiltonians~\cite{VCL09,Kivlichan18,Jiang18}.
	But while these cases are interesting and demonstrate the possibility of fast-forwarding, they are only a handful, and the field of fast-forwarding quantum evolution remains largely unexplored. 
	Characterizing larger classes of Hamiltonians that can be fast-forwarded is crucial.
	
	In this paper, we further explore the theme of fast-forwarding and present new examples that can be fast-forwarded. 
	To this end, we first provide a definition of fast-forwarding that accounts for {\em any} asymptotic complexity improvement of the general case,  going beyond the exponential fast-forwarding studied previously~\cite{AA17}.
	We also consider the problem of fast-forwarding
	in subspaces, which concerns the case where quantum evolution occurs only in a certain subspace of the full Hilbert space.
	This notion of subspace fast-forwarding is useful for simulating physical systems because often
	the evolution occurs in certain subspaces of, for example, low energies or certain preserved symmetries~\cite{csahinouglu2021hamiltonian,TSCT21, su2021nearly}.
	Moreover, since certain quantum systems (e.g., bosonic systems) cannot be directly simulated on a digital quantum computer, we provide a definition 
	of fast-forwarding that is relative to a specified set of observables. 
	The ``elementary gates'' in this case correspond to certain unitary operators that are exponentials of these observables. These definitions are presented in detail in Sec.~\ref{sec:def}.

	Following these definitions, we present several examples of quantum systems that can be fast-forwarded and provide detailed constructions for simulating their time evolution. These are:
	
	\vspace{0.2cm}
	
	\indent {\bf{Block diagonalizable Hamiltonians:}} A class of Hamiltonians that can be exponentially fast-forwarded consists of those Hamiltonians that can be brought into block diagonal form with an efficient quantum circuit if each block can be computed efficiently and fast-forwarded. 
	We investigate fast-forwarding based on this property in Sec.~\ref{subsec:BlockDiagonalization} for spin systems. A particular case,
	which is analyzed in Sec.~\ref{subsubsec:su2hamiltonians},
	is the class of spin Hamiltonians that are invariant under permutations of spins (and thus preserve the magnitude of the total angular momentum). In this case, we achieve exponential fast-forwarding by employing the Schur transform~\cite{BCH06,BCH05}, which block diagonalizes the Hamiltonian in the angular momentum subspaces.

	\indent {\bf{Frustation-free Hamiltonians at low energies:}} In Sec.~\ref{subsec:FrustrationFree}, we investigate subspace fast-forwarding for certain positive Hamiltonians of spin systems, which include those that satisfy a frustration-free property~\cite{PVCW08,BT09}.
	We show that these Hamiltonians can be polynomially fast-forwarded when the input states are guaranteed to be in a certain low-energy subspace.
	We achieve fast-forwarding by using the spectral gap amplification method given in Refs.~\cite{SB13, CS16} in combination with
	quantum phase estimation~\cite{Kit96,CEMM98}.

	\indent {\bf{Fermionic and bosonic Hamiltonians:}} 
	In Sec.~\ref{sec:Fermionic/BosonicModels},
	we study exponential fast-forwarding in quadratic fermionic and bosonic Hamiltonians. We first consider
	a Lie algebraic setting and focus on
	Lie algebraic models of quantum computation, where the quantum gates correspond to certain unitary transformations induced by the algebra.
	In Sec.~\ref{subsec:LieAlgebraDiagonalization} we describe exponential fast-forwarding via Lie-algebra diagonalization. We then apply the technique to formalize a result on exponential fast-forwarding of quadratic fermionic Hamiltonians (Sec.~\ref{subsec:QuadraticFermionicHamiltonians}) and number-conserving quadratic bosonic Hamiltonians (Sec.~\ref{subsec:QuadraticBosonicHamiltonians}).

	\vspace{0.2cm}
	
	Our results address the simulation of quantum systems
	that are abundant in physics. These systems include condensed matter spin models and quantum chemistry and nuclear physics models, such as the well-known Lipkin-Meshkov-Glick model~\cite{lipkin1965validity}, which preserve the magnitude of the total angular momentum. 
	Frustration-free Hamiltonians are also ubiquitous and include the AKLT model~\cite{AKLT88} and, more generally, those ``parent'' Hamiltonians that appear in the context of projected entangled pair states (PEPS)~\cite{PVCW08}.
	While a complexity improvement for simulating these systems under the assumption that the initial state is supported in a low-energy subspace was shown in Ref.~\cite{csahinouglu2021hamiltonian}, that result does not demonstrate fast-forwarding as the quantum complexity is still superlinear in $t$.
	Fermionic and bosonic Hamiltonians are also important, and fast simulation methods for these systems will play an important role in the general problem of simulating quantum field theories~\cite{jordan2012quantum}.
	Our methods exploit different properties of these systems and go beyond the fast-forwarding approach based on diagonalization.

	Last, we note that our definitions immediately lead us to a generalization of the correspondence between fast-forwarding and precise energy measurements discussed in Ref.~\cite{AA17}.
	In Sec.~\ref{sec:energy} we show that, roughly, {polynomial} fast-forwarding is equivalent to {polynomially-precise} energy measurements in qubit systems.

	\section{Concepts and definitions}
	\label{sec:def}
	
	For a given Hamiltonian $H$ and $t>0$, Hamiltonian simulation methods aim at simulating the evolution operator $U(t):=e^{-itH}$. They do this by implementing a sequence of ``elementary'' quantum gates, 
	and the number of these gates determines the quantum complexity of the method. This complexity depends on parameters such as $t$, system size, and precision.
	The fast-forwarding problem is then focused on finding fast ways of simulating $U(t)$,
	where the figure of merit is the quantum complexity and its asymptotic scaling, particularly with $t$. Roughly,
	a Hamiltonian is said to be fast-forwarded if this complexity is sublinear in $t$.

	Several points must be addressed before
	providing precise definitions of fast-forwarding.
	First, the quantum complexity depends on the computational model under consideration, which specifies what the elementary quantum gates are. Usually, these quantum gates are chosen to be compatible with physical implementations.
	For example, in the standard circuit model of quantum computation, elementary quantum gates may correspond to two-qubit gates.  For a fermionic or bosonic model of computation, elementary gates correspond to simple unitary operators induced by corresponding fermionic or bosonic algebras.
	\begin{definition}[Models of quantum computation]
		\label{def:modelsofQC}
		A model of quantum computation can be specified by a set of observables $\mathfrak h =\{O_1,O_2,\ldots\}$
		acting on a Hilbert space $\mathcal H$.
		The elementary quantum gates are unitary transformations expressed as exponentials $e^{-i \theta_l O_l}$, where $\theta_l \in \mathbb R$ and $|\theta_l|\le 1$. A quantum circuit in this model is a sequence of elementary quantum gates, and the number of these gates is the quantum complexity.
	\end{definition}
	
	The condition $|\theta_l|\le 1$ is arbitrary -- i.e., we can bound the phases by a different constant -- but is needed to restrict the power of each quantum gate. 
	This is particularly relevant in bosonic models of computation where, in contrast with certain two-qubit gates (induced by Pauli operators), bosonic gates are not periodic in $\theta_l$.
	\vspace{0.2cm}
	
	Second, known \emph{no-fast-forwarding} results~\cite{BAC07,AA17,HHK+18}
	place a lower bound on the worst-case quantum complexity for simulating classes of Hamiltonians as long as the evolution time satisfies $t \le T$, where $T$ depends on certain problem parameters, such as  the number of qubits or spins. These lower bounds are commonly presented
	as asymptotic scalings.
	If the Hamiltonian is fixed, however, it may be possible to simulate 
	that particular Hamiltonian with quantum complexity that has a very mild (polylogarithmic) dependence on $t$, for example, by exact diagonalization for large $t$. Since the asymptotic scaling of the quantum complexity with $t$ and other parameters is in general of relevance, we will consider sequences of Hamiltonians $\{H_n\}_n$ that depend on a ``system-size'' parameter $n=1,2,\ldots$. For qubit systems, $n$
	is the number of qubits. For fermionic or bosonic systems, $n$ can be the number of modes.
	By allowing the system size to grow, we are able to consider maximum evolution times that increase in order to determine the true asymptotic scaling of the gate complexity.
	For the problem of fast-forwarding quantum evolution, it is then
	important to determine the parameters of Hamiltonian simulation precisely as a function of $n$, which are defined as follows:
	
	\begin{definition} [Hamiltonian-simulation parameters]
		\label{def:simulparamsubsp}
		Let $\{H_n\}_n$ be a sequence of Hamiltonians acting on Hilbert spaces $\{\mc H_n\}_n$.
		The quantum circuits $\{V_n(t)\}_{n,t}$
		simulate the Hamiltonians $\{H_n\}_n$ on subspaces $\{\cS_n \subseteq \cH_n\}_n$
		with parameters $(T(n),\epsilon(n),G(n))$ if the quantum complexity is at most $G(n)$ in a given model of quantum computation 
		and, for each $n$, all $\ket{\psi_n} \in \cS_n$, and all $t\le T(n)$,
		\begin{align}
			\label{eq:FFsubsp}
			\left\|(e^{-i t H_n}\otimes \one_{\mc A} - V_n(t))\ket{\psi_n} \otimes \ket 0_{\mc A}\right\| \le \epsilon(n) \;.
		\end{align}
		The operator $\one_{\mc A}$ is the identity  operator acting on an ancillary register $\cA$ and $\ket 0_{\mc A}$
		is some simple initial state for that register.
	\end{definition}
	
	Throughout this paper, $\|\ket \phi\|$ is the Euclidean norm of the quantum state $\ket \phi$ and $\|A\|$ is the spectral norm of the operator $A$.
	
	\vspace{0.2cm}
	Definition~\ref{def:simulparamsubsp} concerns the
	largest quantum complexity, $G(n)$, 
	for approximating $e^{-itH_n}$ for {\em all} times $t \le T(n)$. While this is not necessary, we will see that understanding the asymptotic behavior of $G(n)$ alone will suffice to 
	classify many classes of Hamiltonians that can be fast-forwarded, including previously known examples and the examples we provide in the following sections. 
	Nevertheless, one may be interested in the ability of fast-forwarding in, for example, other specific regions
	of $t$, in which case the above definition could be adapted to that setting.  
	
	It is important to remark that a circuit $V_n(t)$ might include operations during a pre- or post-processing step. 
	In some cases these steps can be implemented on a classical computer (see the examples in Sec.~\ref{sec:Fermionic/BosonicModels}), but we will assume they are all part of the quantum operation that simulates the Hamiltonian, thereby avoiding any complication in separating quantum and classical complexities.
	Some examples in this paper do require some form of pre-processing, but the complexity of such a step is not dominant in those cases.

	Last, any useful definition of fast-forwarding must capture the ability to simulate a particular Hamiltonian  (or a class of Hamiltonians) with quantum complexity
	that is below a lower bound established 
	for worst-case instances of a class. For example,
	one can establish a no-fast-forwarding theorem for the class of local spin Hamiltonians~\cite{HHK+18} that would apply to the worst case, but there are still
	local Hamiltonians in that class that can be simulated more rapidly with quantum complexity that is sublinear in $t$, such as XY Ising models~\cite{VCL09}. The no-fast-forwarding line, which we define below, will play an important role in a definition of fast-forwarding.
	
	\begin{definition}[No-fast-forwarding line]
		\label{def:non-fastforwarding_line}
		Let $\{\mc C_n\}_n$ denote classes of Hamiltonians, i.e. $\mc C_n$ is a subset of the Hermitian operators acting on $\mc H_n$, and let $T(n) \ge 0$ and $\epsilon(n) \ge 0$ be functions of $n$.
		The no-fast-forwarding line $l(n)$ with respect to these classes is the function
		\begin{align}
			l(n) = \max_{H_n\in \mc C_n}\min_{\{V_n(t)\}_t}G(n) \; ,
		\end{align}
		where the quantum circuits $\{V_n(t)\}_t$ simulate the Hamiltonian $H_n \in \mc C_n$ with parameters $(T(n),\epsilon(n),G(n))$.
	\end{definition}
	
	That is, for a given $n$, the no-fast-forwarding line is the minimum quantum complexity required for simulating every Hamiltonian in a given class and quantum computational model. This line depends on the maximum simulation time $T(n)$; however, in our analyses and examples $T(n)$ is a fixed function of $n$ and we do not need to consider this dependence explicitly.

	If a particular sequence of Hamiltonians of those classes can be simulated with less quantum complexity, i.e. asymptotically less than $l(n)$, then we will claim that such a Hamiltonian sequence can be fast-forwarded. In more detail, we define
	fast-forwarding as follows:
	
	\begin{definition}[Fast-forwarding]
		\label{def:general_fast-forwarding}
		Let $\{\mc C_n\}_n$ denote classes of Hamiltonians, i.e. $\mc C_n$ is a subset of the Hermitian operators acting on $\mc H_n$, and let $T(n)\ge 0$ and $\epsilon(n) \ge 0$ be functions of $n$.
		The quantum circuits $\{V_n(t)\}_{n,t}$ are said to be 
		$(T(n), \epsilon(n), G(n))$-fast-forwarding a Hamiltonian sequence $\{H_n \in \mc C_n\}_n$
		on subspaces $\{\cS_n\}_n$
		if the following hold:
		
		\begin{enumerate}
			\item \sloppy The quantum circuits $\{V_n(t)\}_{n,t}$ simulate the Hamiltonians $\{H_n\}_n$ with parameters $(T(n),\epsilon(n),G(n))$, and
			
			\item 
			\begin{align}
				\lim_{n \rightarrow \infty}\frac{G(n)}{l(n)} =0 \;.
			\end{align}
		\end{enumerate}
	\end{definition}

	Our definition of fast-forwarding is relative to the classes of Hamiltonians $\{\mc C_n\}_n$, and this is needed because certain classes are more difficult to simulate than others. If, however, those classes are chosen to be too restrictive, then the definition will not capture
	the ability to fast-forward a Hamiltonian in an interesting class. For example, a diagonal qubit Hamiltonian $H_n$ may not be fast-forwarded relative to a class $\{\mc C_n\}_n$ of diagonal qubit Hamiltonians despite having an algorithm of polylogarithmic complexity in $t$ because every Hamiltonian in $\cC_n$ can be simulated with such complexity. But if $\{\mc C_n\}_n$ are more general, not necessarily diagonal qubit Hamiltonians, and $H_n \in \mc C_n$ is diagonal, then the definition will serve its purpose. Likewise, we will require $\{\mc C_n\}_n$ and $\{ H_n\}_n$ to consist of Hamiltonians that are normalized in some way. 
	Otherwise, $l(n)$ and $G(n)$ could be artificially large or small, respectively, implying fast-forwarding under our definition.

	Definition~\ref{def:general_fast-forwarding} also depends on a Hamiltonian
	sequence $\{H_n\}_n$ with no explicit requirement
	that these Hamiltonians are related in some way.
	Nevertheless, the definition becomes useful
	when the sequence describes the ``same'' type of Hamiltonians for different system sizes $n$.
	For example, $H_n$ may be constructed from $H_{n-1}$ by adding a term that involves a new subsystem; e.g., the $n$-th spin or the $n$-th fermionic or bosonic mode.
	We sometimes refer to such sequences as \emph{meaningful sequences of Hamiltonians} and define them more precisely for spin and fermionic or bosonic systems in Secs.~\ref{sec:spinsystems} and~\ref{sec:Fermionic/BosonicModels}, respectively.
	Similarly, while no requirement is set for the subspaces $\{\cS_n\}_n$, the definition will be useful when $\cS_n$ is related to $\cS_{n-1}$ in some way (e.g., $\cS_n \equiv \mc H_n$).
	For simplicity, we will be concerned with the simulation of Hamiltonians to constant error $\epsilon(n)=\epsilon$, but our results can be easily generalized to other cases.
	
	\vspace{0.2cm}
	
	It is generally difficult to find the no-fast-forwarding line $l(n)$.
	To show that a Hamiltonian can be fast-forwarded, we often show that it can be simulated with asymptotic quantum complexity that grows slower in $t$ than a lower bound $l'(n,t)$ of $l(n)$. These lower bounds, also known as no-fast-forwarding theorems, are known for some classes of Hamiltonians.
	Reference~\cite{HHK+18}, for example, provides this lower bound for the class of geometrically-local qubit Hamiltonians, although these Hamiltonians are time-dependent. 
	In this case, $l'(n,t)=\tilde\Omega(n t)$
	for $t \le T(n)=4^n$ and thus $l(n)=\tilde \Omega(n 4^n)$.
	Related is a result of Ref.~\cite{AA17} that
	strongly suggests that no exponential fast-forwarding is possible for all $n$-qubit Hamiltonians that are 2-sparse,
	where the matrix of $H_n$ has at most two nonzero elements per row/column. Other known lower bounds apply
	to the setting in which Hamiltonians can be accessed via a black-box unitary~\cite{BAC07}.
	These results suggest that for some classes of Hamiltonians analyzed in this paper, the no-fast-forwarding line satisfies $l(n)= \Omega(nT(n))$, where $T(n)$ is exponential in $n$ (e.g., $T(n)=4^n$). When the no-fast-forwarding line has this form, Def.~\ref{def:general_fast-forwarding} coincides with the intuitive notion of fast-forwarding as a simulation algorithm achieving complexity scaling sublinearly in the evolution time.

	The various types of fast-forwarding can be characterized by the asymptotic behavior of $G(n)$ and $l(n)$. For example, for the exponential fast-forwarding problem analyzed in Ref.~\cite{AA17}, a class of Hamiltonians can be simulated with quantum complexity $G(n) = \cO ({\rm poly}(n))$
	while $l(n)$ (or $T(n)$) is exponential in $n$. 
	Our definition, however, allows for more general cases, including polynomial fast-forwarding. This occurs when, in particular, both $G(n)$ and $l(n)$ are  exponential in $n$, but $G(n)/l(n) \rightarrow 0$. For example, a Hamiltonian simulation method of quantum complexity that scales as $n t^\alpha$ for $\alpha<1$ and all $t \le T(n)= 4^n$,
	will imply polynomial fast-forwarding relative to a class of Hamiltonians with no-fast-forwarding line of the form $l(n)= \Omega(n4^n)$.
	In Fig.~\ref{fig:FFDef} we sketch  examples of fast-forwarding and no-fast-forwarding
	using our definitions, where all the relevant quantities are increasing in $n$ or $t$.

	\begin{figure} [htbp]
		\centering
		\includegraphics[width=1.\linewidth]{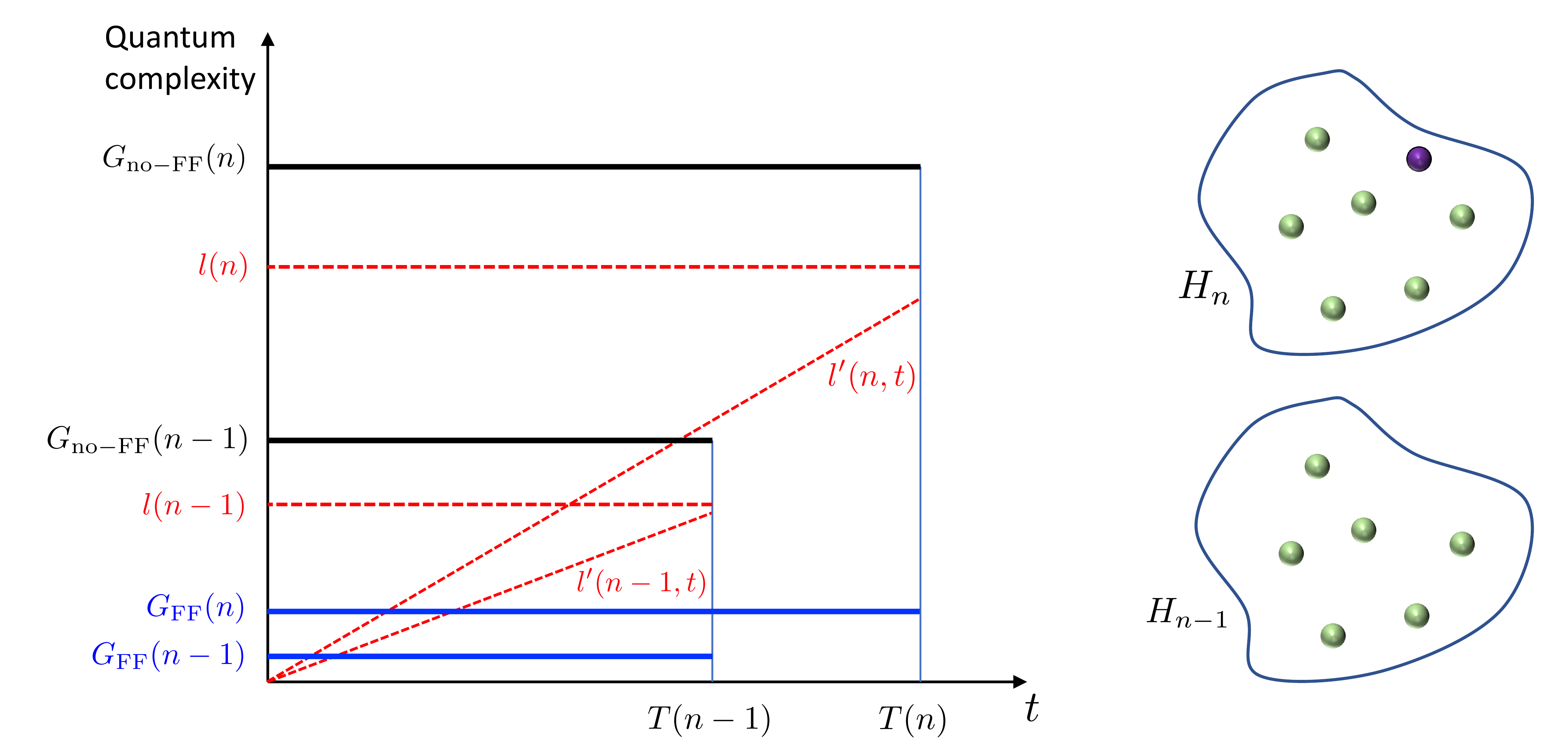}
		\caption{\label{fig:FFDef}Examples of fast-forwarding and no-fast-forwarding according to Def.~\ref{def:general_fast-forwarding}. For fast-forwarding, the quantum complexity $G_{FF}(n)$ crosses the lower bound $l'(n,t)$ and lies under the no-fast-forwarding line $l(n)$, and $G_{{\rm FF}}(n)/l(n)\rightarrow 0$. For no-fast-forwarding, the quantum complexity $G_{{\rm no-FF}}(n)$ lies above $l(n)$. Exponential or polynomial fast-forwarding is obtained depending on the asymptotic behavior of $G(n)$ and $l(n)$. The Hamiltonians $\{H_n\}_n$ belong to classes $\{\cC_n\}_n$ and model quantum systems of different sizes $n$.
		}
	\end{figure}

	In the following, we assume $\hbar=1$ and do not assign units to Hamiltonians or time,
	which is common practice for analyzing quantum simulation algorithms. Nevertheless, if one is interested in the resulting complexity in an actual application where units are considered, this will easily follow from our results. 
	We also use $\cA$, $\cA'$, etc., to denote ancillary systems, which include {\em all} additional systems required by a particular Hamiltonian simulation approach, such as those needed to implement certain classical computations reversibly.

	\section{Fast-forwarding of local spin systems}
	\label{sec:spinsystems}

	In this section, we consider fast-forwarding relative to the class of local spin systems, where $n$ spins are
	located at the vertices of lattices $\Lambda_n$, $n=1,2,\ldots$.
	Each spin is of dimension $d$, and for each $n$, the Hilbert space is $\mc H_n = \left(\mathbb{C}^{d}\right)^{\otimes n}$.
	For each system there is a Hamiltonian $H_n$ acting on $\cH_n$ that describes the interactions. This Hamiltonian is a sum of (local) Hermitian terms; that is, $H_n= \sum_{X \subset \Lambda_n} h_X$, where $h_X$ describes the interactions of spins in a subset $X \subset \Lambda_n$. 
	We detail the structure of these Hamiltonians by assuming them to be $k$-local and of degree $g$.
	This means that each subset $X$ involves at most $k$ spins, and each spin is involved in at most $g$ of the subsets $X$. For simplicity, we take $k$ and $g$ to be constant and also assume $\|h_X\| \le 1$, for all $X$. 
	The no-fast-forwarding line is assumed to be $l(n)=\Omega(nT(n))$, where $T(n)=d^{2n}$.

	The meaningful Hamiltonian sequences $\{H_n\}_n$ are such that $H_n = H_{n-1} \otimes \one_n + v_n$, where $v_n$ is a sum of interaction terms that act  non-trivially on the new ($n$-th) spin and other spins of the system.
	The Hilbert spaces are such that $\cH_n=\cH_{n-1} \otimes \mathbb C^d$ and $\one_n$ is the identity operator on $\mathbb C^d$, the space associated with the new spin. 
	We are interested in the complexity of simulating spin Hamiltonians in the standard gate-based model, where elementary quantum gates are two-qubit gates.
	Nevertheless, some results of this section apply to more general quantum systems that are not necessarily spin systems.

	The first example of fast-forwarding of a spin system is described in Sec.~\ref{subsec:BlockDiagonalization}, where we study Hamiltonians with certain block diagonal structure, which occurs due to a symmetry of the system.
	Roughly, if a Hamiltonian can be taken to a block diagonal form via a unitary that admits an efficient quantum circuit implementation, and the block diagonal Hamiltonian can be simulated with complexity sublinear in $t$, then the original Hamiltonian can be fast-forwarded.
	This occurs, for example, when the Hamiltonian is invariant under the permutation of spins, such as for the well-known Lipkin-Meshkov-Glick model~\cite{lipkin1965validity}, and we show that this case can be exponentially fast-forwarded.
	Permutation-invariant spin Hamiltonians have also been used in quantum computing for noise protection~\cite{zanardi1999symmetrizing} and optimization~\cite{cook2020quantum}.

	The second example is described in Sec.~\ref{subsec:FrustrationFree}, where we study the so-called parent or frustration-free Hamiltonians in the context of subspace fast-forwarding.
	We show that these Hamiltonians can be polynomially fast-forwarded if the subspaces $\cS_n$ are those of  ``low'' energies.
	The actual level of fast-forwarding depends on the parameters of the problem, especially on a low-energy cutoff $\Delta$, which might depend on the evolution time $t$.

	In the following, we  drop the subscript $n$
	from the Hamiltonians and Hilbert spaces when these are clear from context.

	\subsection{Block diagonalization}
	\label{subsec:BlockDiagonalization}
	
	Any $n$-spin Hamiltonian $H$ can be decomposed based on its action on invariant subspaces as
	\begin{align}
		H = \bigoplus_\mu \left.H\right|_{\cI_\mu} \; ,
	\end{align}
	where $\mu$ may be associated with some conserved quantity (quantum number), $\left.H\right|_{\cI_\mu}$ is the restriction of the Hamiltonian to the subspace $\cI_\mu \subseteq \cH$, and $\cH = \bigoplus_\mu \cI_\mu$.
	To fast-forward $H$, we are particularly interested in cases where the evolution under each $\left.H\right|_{\cI_\mu}$ can be fast-forwarded. More precisely, let $U$ be a block diagonalizing unitary operation or quantum circuit that
	performs the mapping
	\begin{align}
		\label{eq:block-diagonalizing}
		U \left( H \otimes \ketbra 0_{\cA} \right) U^\dagger = \sum_\mu \left.H'\right|_{\cI'_\mu} \otimes \ketbra \mu_{\cA'} \; ,
	\end{align}
	where $\cA,\cA'$ are some ancillary systems and the basis states $\{\ket{\mu}_{\cA'}\}_\mu$ encode the values of $\mu$.
	In general, $\cI'_\mu$ are subspaces of a different Hilbert space $\cH'$, so the Hamiltonians 
	$\left.H'\right|_{\cI'_\mu}$ may be different than $\left.H\right|_{\cI_\mu}$, but we require $\dim \cI_\mu = \dim \cI'_\mu$.
	Note that the tensor product decompositions on the left and right hand sides of Eq.~\eqref{eq:block-diagonalizing} may also be different.
	Consider now another unitary
	\begin{align}
		\label{eq:U'def}
		U'(t)=\sum_\mu U_\mu(t) \otimes \ketbra \mu_{\cA'}
	\end{align}
	that approximates the evolution under each $\left.H'\right|_{\cI'_\mu}$ as
	\begin{align}\label{eq:BlockUnitary}
		\left\| \left(U_\mu(t) - e^{-i t H'|_{\cI'_\mu}} \right) \ket \psi  \right\| \le \epsilon \;,
	\end{align}
	for all $\ket \psi \in \cI'_\mu$. 
	We obtain:
	\begin{theorem}
		Let $\epsilon > 0$, $T>0$, $g$ be the quantum complexity of $U$ in Eq.~\eqref{eq:block-diagonalizing}, and $g'(t,\epsilon)$ be an upper bound on the quantum complexity of  $U'(t)$ in Eq.~\eqref{eq:U'def}, such that Eq.~\eqref{eq:BlockUnitary} holds. 
		Assume that $g'(T,\epsilon) \ge g'(t,\epsilon)$ for all $t \in [0,T]$. 
		Then, there exist quantum circuits $\{V(t)\}_t$ that simulate $H$ on $\cH$ with parameters $(T,\epsilon,2g+ g'(T,\epsilon) )$ for all $0 \le t \le T$.
	\end{theorem}
	
	\begin{proof}
		For each $t \le T$, the quantum circuit $V(t)=U^\dag U'(t)U$ of Fig.~\ref{fig:block diagonal evolution} can be used to approximate $e^{-itH}$. That circuit implements $U$ and $U^\dagger$, adding $2g$ to the quantum complexity, and also implements $U'(t)$ that has quantum complexity bounded by $g'(T,\epsilon)$. We can bound the simulation error as follows:
		\begin{align}
			\| (e^{-i t H} \otimes \one_\cA - V(t)) \ket \psi \otimes \ket 0_\cA \| &= \| (e^{-i t (H \otimes \ketbra 0_\cA)}  - V(t)) \ket \psi \otimes \ket 0_\cA\| \\
			& =  \| (\sum_\mu (e^{-i t H'|_{\cI'_\mu}} - U_\mu(t)) \otimes \ketbra \mu_{\cA'} )\ket \phi  \|
			\\
			&\le \epsilon \;,
		\end{align}
		where $\ket \psi \in \cH$ is any state and $\ket \phi=U \ket \psi \otimes \ket 0_\cA$.

		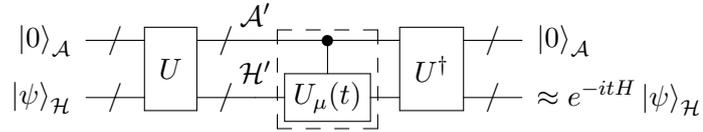
\begin{figure} [htbp]
			\centerline{\Qcircuit @C=1em @R=.7em {
					\lstick{\ket 0_{\cA}} & {/}\qw & \multigate{1}{U} & {/}\qw & \ustick{\cA'}\qw & \ctrl{1} & \multigate{1}{U^\dag} & {/}\qw & \rstick{\ket 0_{\mc A}}\qw\\
					\lstick{\ket \psi_{\cH}} & {/}\qw & \ghost{U} & {/}\qw & \ustick{\cH'}\qw & \gate{U_\mu(t)} & \ghost{U^\dag} & {/}\qw & \rstick{\approx e^{-itH}\ket\psi_{\cH}}\qw \gategroup{1}{6}{2}{6}{0.5em}{--}
			}}
			\caption{\label{fig:block diagonal evolution}The quantum circuit $V(t)$ simulating the time evolution of a block diagonalizable Hamiltonian. The dashed box denotes the unitary $U'(t)$, which implements $U_\mu(t)$ conditional on $\ket{\mu}_{\cA'}$.
			}
		\end{figure}
		
	\end{proof}

	Fast-forwarding using this approach is possible if the quantum complexity of $V(t)$ is sublinear in $t$ for a sufficiently large range of values of $t$.
	Thus, the ability to fast-forward block diagonalizable Hamiltonians depends on the complexity of the various procedures in Fig.~\ref{fig:block diagonal evolution}. Of particular interest
	are those Hamiltonians where the unitary $U$ in Eq.~\eqref{eq:block-diagonalizing}
	has quantum complexity $g={\rm poly}(n)$ and $U'(t)$ in Eq.~\eqref{eq:U'def} has quantum complexity $g'(T,\epsilon)={\rm poly}(n,\log(T))$, for constant error $\epsilon$.
	These Hamiltonians can be exponentially fast-forwarded. They include  Hamiltonians that can be efficiently (block) diagonalized, where $\dim \cI'_\mu$ is polynomial in $n$ and, for each $\mu$, the Hamiltonians $H'|_{\cI'_\mu}$ can be efficiently computed and simulated, as in the example of Sec.~\ref{subsubsec:su2hamiltonians} below.
	
	Note that implementing $U'(t)$ requires implementing
	the unitaries $U_\mu(t)$ conditional on the ancillary state $\ket \mu_{\cA'}$. 
	There are at most $\dim \cH=d^n$ different values of $\mu$, and thus at most $d^n$ unitaries $U_\mu(t)$. The quantum complexity of $U'(t)$ can then be exponential in $n$ in the worst case, but as we discussed, we are interested in cases where this complexity is only ${\rm poly}(n)$.  A standard approach to achieve this complexity in some cases is by having access to a classical algorithm that, on any input $\mu$, provides a classical description of the quantum circuit $U_\mu(t)$ in time ${\rm poly}(n)$. The classical circuit can be efficiently implemented on a quantum computer using reversible quantum gates, which we denote by $C$; see Fig.~\ref{fig:block diagonal controlled unitary}. The unitary $U'(t)$ can then be efficiently implemented by applying the corresponding quantum gates of $U_\mu(t)$ conditional on the ancillary register state that encodes the quantum circuit. The last step is to apply the inverse of $C$ to reset the ancillary state. 
	
	\vspace{0.4cm}
	
	\begin{figure} [htbp]
		\centerline{\Qcircuit @C=1em @R=.7em {
				\lstick{\ket \mu_{\mc A'}} & {/}\qw & \multigate{1}{C^{\;}} & \qw & \multigate{1}{C^\dagger} & \rstick{\ket \mu_{\mc A'}}\qw\\
				\lstick{\ket 0_{\mc A''}} & {/}\qw & \ghost{C^{\;}} & \ctrl{1} & \ghost{C^\dagger} & \rstick{\ket 0_{\mc A''}}\qw\\
				\lstick{\ket {\varphi}_{\cH'}} & {/}\qw & \qw & \gate{U_\mu(t)} & \qw & \rstick{U_\mu(t)\ket {\varphi}_{\cH'}}\qw
		}}
		\caption{\label{fig:block diagonal controlled unitary}
			Efficient implementation of $U'(t)$ using $C$, a reversible version 
			of a classical circuit that outputs a description of $U_\mu(t)$ on input $\mu$. The ancillary system is comprised of two registers: $\cA'$ needed to encode $\mu$ and $\cA''$ needed to encode the description of $U_\mu(t)$.}
	\end{figure}
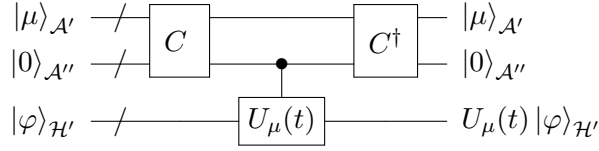

	\subsubsection{Permutation-invariant spin Hamiltonians}\label{subsubsec:su2hamiltonians}
	
	We apply the results of Sec.~\ref{subsec:BlockDiagonalization}
	to fast-forward Hamiltonians that are invariant under permutations of spins. 
	This setting is relevant for studying nuclear and condensed matter systems, where the well-known Lipkin-Meshkov-Glick model is a prime example~\cite{lipkin1965validity}.
	We use the block diagonal structure of these Hamiltonians to show that their time evolution can be exponentially fast-forwarded. In particular, the permutation-invariance ensures the following:
	
	\begin{enumerate}
		\item The Schur transform $U_{\rm Sch}$, which is a unitary transformation, takes the Hamiltonian to a block diagonal form according to the decomposition $\mc H= \bigoplus_\lambda \cQ_\lambda  \otimes \cP_\lambda$, where $\cQ_\lambda$ and 
		$\mathcal P_\lambda$ are the irreducible representations of the unitary group $U(d)$ and symmetric group $S_n$, respectively, labeled by the Young diagrams $\lambda\in \operatorname{Part}(n,d)$.
		An efficient quantum algorithm to implement the Schur transform $U_{\rm Sch}$ is shown in Refs.~\cite{BCH06,BCH05} and analyzed in detail for the case $d=2$ in Appendix~\ref{appx:schur}.
		
		\item Each subspace $\cQ_\lambda$ is of dimension $\poly(n)$ and the Hamiltonian acts trivially on the tensor factors $\mc P_{\lambda}$, i.e., for all $\lambda$ and $\ket{q}\otimes\ket{p}\in \cQ_\lambda\otimes\cP_\lambda$, the Hamiltonian acts as $H(\ket{q}\otimes\ket{p}) = (H_\lambda\ket{q}) \otimes \ket p$.
	\end{enumerate}
	
	Thus, the block diagonalizing quantum circuit in Eq.~\eqref{eq:block-diagonalizing} is $U=U_{{\rm Sch}}$, and the blocks or subspaces are labeled by quantum numbers $\mu=(\lambda,p)$ associated with invariant subspaces $\mc I_{\mu}=\cQ_\lambda \otimes \bb C\ket p$.
	Because the dimensions of the blocks are polynomial, we obtain an efficient circuit $U'(t)$ for simulating the block diagonal Hamiltonians in the Schur basis, obtaining the following result for permutation-invariant Hamiltonians:
	
	\begin{theorem}[Exponential fast-forwarding of permutation-invariant Hamiltonians]
		\label{thm:sud}
		Let $\{\cC_n\}_n$ denote the classes of $k$-local, $n$-spin Hamiltonians of dimension $d$ acting on spaces $\cH_n =\left(\mathbb C^d\right)^{\otimes n}$ and $\{H_n \in \cC_n\}_n$ be a sequence of permutation-invariant Hamiltonians. Then, for given $\epsilon > 0$, $T(n) = d^{2n}$, there exist quantum circuits $\{V_n(t)\}_{n,t}$ in the standard gate model that are $(T(n),\epsilon,G(n))$-fast-forwarding the Hamiltonians $\{H_n\}_n$ on the spaces $\{\cH_n\}_n$, where $G(n)=\poly(n)$. In particular, if $l(n)=\Omega(nT(n))$, this implies exponential fast-forwarding.
	\end{theorem}
	
	\begin{proof}
		For a given $n$, Fig.~\ref{fig:su2 simulation} shows the circuit used to simulate the Hamiltonian $H=H_n$ for time $t$. The Schur transform $U_{{\rm Sch}}$ block diagonalizes the Hamiltonian $H$ as follows.
		Consider an $n$-spin input state $\ket x\in \cH$ along with an ancillary state $\ket 0_\cA$.
		The unitary $U_{{\rm Sch}}$ is a basis transformation according to the Schur-Weyl duality. In the Schur basis, basis states $\ket\lambda$ of a register $\cA'_1$ label the Young diagram while each basis state $\ket p\in\mc P'_\lambda$
		of a register $\cA'_2$ labels a basis vector of the corresponding symmetric group irreducible representation, 
		and $\cA'=\cA'_1 \cup \cA'_2$ in the notation of Sec.~\ref{subsec:BlockDiagonalization}; see Eq.~\eqref{eq:block-diagonalizing} and Fig.~\ref{fig:block diagonal evolution}.
		Also, the basis states $\ket q$ satisfy
		$\ket q \in \mc Q'_{\lambda} \subseteq \cH'$.
		The input state is then mapped to a linear combination of the Schur basis states $\{\ket q_{\cH'}\otimes \ket\lambda_{\cA'_1} \otimes \ket p_{\cA'_2} \}_{q,\lambda, p}$. 
		As we have argued, $H$ is block diagonal in the Schur basis, only acting on $\cH'$ conditioned on $\ket{\lambda}_{\cA'_1}$. 
		After implementing the operator $U'(t)=\sum_\lambda U_\lambda(t) \otimes ( \ketbra \lambda_{\cA'_1} \otimes \one_{\cA'_2})$, where $U_\lambda(t)=e^{-it \left .H \right|_{\cQ'_\lambda}}$, we transform back to the original basis.
		
		\begin{figure} [htbp]
			\centerline{\Qcircuit @C=1em @R=1.5em {
					\lstick{\ket 0_{\mc A}} & {/}\qw & \multigate{2}{U_{{\rm Sch}}} & {/}\qw & \qw & \ustick{\ket \lambda_{\cA'_1}}\qw & \qw & \qw & \ctrl{2}& \ustick{\! \! \! \! \! \! \! \! \! \! \! \! \! \! \cA'}\qw & \qw &\multigate{2}{U_{{\rm Sch}}^\dagger} & {/}\qw & \rstick{\ket 0_{\mc A}} \qw \\
					& & \push{\rule{3em}{0em}} & {/}\qw & \qw & \ustick{\ket p_{\cA'_2}}\qw & \qw & \qw & \qw & \qw & \qw &\ghost{U_{{\rm Sch}}^\dagger} & \\
					\lstick{\ket x_{\cH}} & {/}\qw  &\ghost{U_{{\rm Sch}}} & {/}\qw & \qw & \ustick{\ket q_{\cH'}}\qw & \qw & \qw & \gate{e^{-it\left.H\right|_{\mc Q'_\lambda}}}& \qw & \qw &\ghost{U_{{\rm Sch}}^\dagger} & {/}\qw & \rstick{e^{-itH}\ket x_{\cH}} \qw \gategroup{1}{9}{2}{9}{1.5em}{.}
			}}
			\caption{\label{fig:su2 simulation}The quantum circuit that simulates the time evolution of a permutation-invariant Hamiltonian. The ancillary register $\cA'$ appearing in Eq.~\eqref{eq:block-diagonalizing} contains two registers $\cA'_1$ and $\cA'_2$ in this case. The state after the basis transformation $U_{\rm Sch}$ is in general a linear combination of Schur basis states $\{\ket q_{\cH'} \otimes \ket\lambda_{\cA'_1} \otimes \ket p_{\cA'_2} \}_{q,\lambda, p}$.}
		\end{figure}
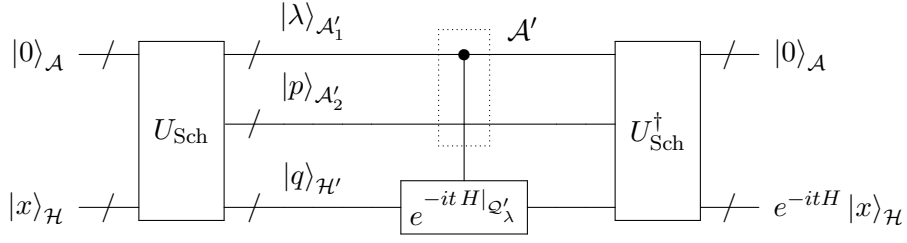
		
		We now consider the complexity of simulating $H$
		when $t\le T(n)=d^{2n}$ and $\epsilon>0$ is constant.
		We allow $\epsilon/3$ error for each of the three components $U_{\rm Sch}$, $e^{-itH|_{\cQ'_\lambda}}$, and $U_{\rm Sch}^\dag$.
		As described in Ref.~\cite{BCH05}, $U_{\rm Sch}$ (and $U_{\rm Sch}^\dag$) can be implemented with complexity $n\poly(\log n,d,\log \frac 1 \epsilon)=\poly(n)$.
		In the Schur basis, each subspace $\mc Q'_\lambda$ has dimension given by the number of semistandard Young tableaux of shape $\lambda$ and maximum entry $d$, which is $\cO(n^{d^2})$\footnote{An explicit formula is given by $\dim \mc Q'_\lambda = \dim \mc Q_\lambda = \frac{\prod_{1\le i<j\le d}\lambda_i-\lambda_j+j-i}{\prod_{m=1}^d (m-1)!}$.}.
		Thus, each block $e^{-itH|_{\cQ'_\lambda}}$ of the unitary operator $U'(t)$ has size polynomial in $n$ and can be decomposed into a product of ${\rm poly}(n)$ two-level unitaries using QR decomposition by Givens rotations, as described in Section 4.5.1 of Ref.~\cite{NC01}.
		The circuit $U'(t)$ consists of these two-level unitaries controlled on the $\cA'_1$ register.
		By classically calculating a description of the quantum circuit for $\lambda$ on another ancillary register and applying the conditional two-level unitaries  as in Fig.~\ref{fig:block diagonal controlled unitary}, we obtain a circuit implementing $U'(t)$.
		Alternatively, because there are at most $|\operatorname{Part}(n,d)|=\binom{n+d-1}{d-1}=\cO(n^{d-1})$ labels $\lambda$, we could also evolve the blocks conditioned on the each value of $\lambda$ separately.
		In either implementation, the circuit for $U'(t)$ has quantum complexity that is polynomial in $n$.
	\end{proof}
	
	We can obtain a more precise expression for the simulation complexity by fixing specific encodings for the different registers in the Schur basis. For simplicity, we demonstrate the result for qubit systems:
	
	\begin{theorem}[Exponential fast-forwarding of permutation-invariant  qubit Hamiltonians]
		\label{thm:su2}
		Let $\{\cC_n\}_n$ denote the classes of $k$-local, $n$-qubit Hamiltonians acting on spaces $\cH_n=\left(\mathbb C^2\right)^{\otimes n}$ and
		$\{H_n \in \cC_n \}_n$ be a sequence of permutation-invariant qubit Hamiltonians. Then, for given $\epsilon > 0$, $T(n)=2^{2n}$, there exist quantum circuits $\{V_n(t)\}_{n,t}$ in the standard gate model that are $(T(n),\epsilon,G(n))$-fast-forwarding
		the Hamiltonians $\{H_n \in \cC_n \}_n$ on the spaces $\{\cH_n\}_n$, 
		where $G(n)=\cO(n^3)$.
		In particular, if $l(n)= \Omega(n T(n))$, this implies exponential fast-forwarding.
	\end{theorem}
	
	\begin{proof}
		For $d=2$, a partition $\lambda=(\lambda_1,\lambda_2)$ into at most two parts with $\lambda_1+\lambda_2=n, \lambda_1\ge\lambda_2\ge 0$ can equivalently be described by providing $J=\frac{\lambda_2-\lambda_1}{2}$, which is the total angular momentum labeling representations of $SU(2)$.
		
		In Appendix~\ref{appx:schur}, we explicitly describe the circuit for the Schur transform for the qubit case, following Ref.~\cite{BCH06}, and using a unary encoding of the registers that encode $\lambda$ (or $J$) and $q$. That is, for a given basis state $\ket J$ or $\ket q$, all qubits of the register are $\ket 0$ except the one corresponding to the value of $J$ or $q$, which is $\ket 1$. The gate complexity of the circuit $U_{\rm Sch}$ is $\cO(n^3)$.
		
		The dimensions of the irreducible representations of $SU(2)$ are $\dim \mc Q_J=2J+1 \le n+1$, so the QR decomposition expresses $U_J(t)$ as the product of $\mc O(n^2)$ two-level operations. The unary encoding ensures that each two-level unitary, acting on a subspace spanned by $\ket q, \ket{q+1}$, is a two-qubit gate with support on the corresponding qubits. We implement the block diagonal unitary $U'(t)$ by evolving the system in each irreducible representation separately, as shown in Fig.~\ref{fig:block evolution}. By encoding the $\ket J$ register in unary, the unitaries acting on the $\ket q$ register are controlled on single qubits in the $\ket J$ register, resulting in a constant cost per two-level unitary. Because there are $\mc O(n)$ possible values of $J$, this gives a total of $\mc O(n^3)$ two-qubit elementary gates to implement $U'(t)$. Adding the cost of the quantum complexity of $U_{{\rm Sch}}$, we can simulate the evolution of the Hamiltonian $H$ with $\cO(n^3)$ elementary gates.
		
		\begin{figure} [htbp]
			\centerline{\Qcircuit @C=1em @R=1em {
					\lstick{J=0} & \qw & \ctrl{5} & \qw & \qw & \qw & \qw \\
					\lstick{J=\frac 1 2} & \qw & \qw & \ctrl{4} & \qw & \qw & \qw \\
					\lstick{J=1} & \qw & \qw & \qw & \ctrl{3} & \qw & \qw \\
					\lstick{J=\frac 3 2} & \qw & \qw & \qw & \qw & \ctrl{2} & \qw \\
					\\
					\lstick{\ket q} & {/}\qw & \gate{U_{J=0}(t)} & \gate{U_{J=1/2}(t)} & \gate{U_{J=1}(t)} & \gate{U_{J=3/2}(t)} & \qw
			}}
			\caption{\label{fig:block evolution}
				The quantum circuit $U'(t)$ implementing the evolution of the block diagonal Hamiltonian in the Schur basis for $n=3$ qubits. This corresponds to the middle block of Fig.~\ref{fig:su2 simulation}. Each controlled $U_J(t)$ operation consists of $\mc O(n^2)$ elementary gates.
			}
		\end{figure}
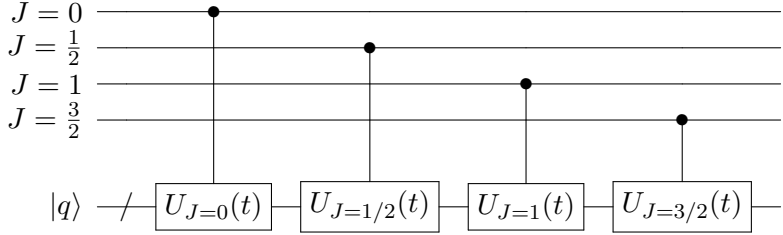
		
	\end{proof}
	
	Examples of permutation-invariant Hamiltonians are vast.
	One characterization is that they are in the universal enveloping algebra $U(\mf{su}(d))$:
	Consider a Hamiltonian $H$ with terms that are products of $\mf{su}(d)$ operators $A$ acting on the representation space $\mc H=\left(\bb C^d\right)^{\otimes n}$ as
	\begin{align}
		A\ket{x_1\dots x_n} = \sum_{i=1}^n \ket{x_1}\otimes\dots\otimes A\ket{x_i}\otimes\dots\otimes\ket{x_n}\;.
	\end{align}
	For example, the $\mf{su}(2)$ operators are spanned by angular momentum operators
	\begin{align}
		J_x = \frac{1}{2}\sum_i X_i\;,\quad J_y = \frac{1}{2}\sum_i Y_i\;,\quad J_z = \frac{1}{2}\sum_i Z_i\;,
	\end{align}
	where each qubit of the Hilbert space is identified with the standard representation $\mc Q_{1/2}$ of $\mf{su}(2)$.
	Since $H$ is manifestly permutation-invariant, any $U(\mf{su}(d))$ Hamiltonian can be fast-forwarded using Thm.~\ref{thm:sud}.
	Conversely, this example is the most general case because up to an additive constant, every permutation-invariant Hamiltonian is an element of $U(\mf{su}(d))$.
	Indeed, by Schur-Weyl duality, the commutant of the image of the symmetric group $S_n$ in $\left(\bb C^d\right)^{\otimes n}$ is spanned by elements of the form $A^{\otimes n}$ for a $A\in M_d(\bb C)$, which are generated by $\mf{gl}(d)$. The Hermitian elements are the $U(\mf{su}(d))$ operators up to an additive constant.

	A famous example of a $U(\mf{su}(2))$ Hamiltonian is the Lipkin-Meshkov-Glick (LMG) model that appears in nuclear physics~\cite{lipkin1965validity}. The LMG model describes a system of $n$ fermions which can each be spin up or spin down. 
	The Hamiltonian is given by
	\begin{align}
		H = \frac{1}{2}\sum_{i,\sigma}\sigma c_{i\sigma}^\dagger c_{i\sigma}^{\;} + \frac{V}{2n}\sum_{i,i',\sigma}c_{i\sigma}^\dagger c_{i'\sigma}^\dagger c_{i'\bar\sigma}^{\;} c_{i\bar\sigma}^{\;} + \frac{W}{2n}\sum_{i,i',\sigma}c_{i\sigma}^\dagger c_{i'\bar\sigma}^\dagger c_{i'\sigma}^{\;} c_{i\bar\sigma}^{\;}\;.
	\end{align}
	Here, $\sigma=\uparrow,\downarrow$ corresponds to the two possible spins with $\bar\sigma$ the opposite spin, $i$ runs from 1 to $n$, and the constants $V$ and $W$ describe the interaction strengths.
	By defining the operators
	\begin{align}
		J_+ = \frac{1}{\sqrt 2}\sum_i c_{i\uparrow}^\dagger c_{i\downarrow}^{\;}\;, \quad
		J_- = \frac{1}{\sqrt 2}\sum_i c_{i\downarrow}^\dagger c_{i\uparrow}^{\;}\;, \quad
		J_z = \frac{1}{2}\sum_{i,\sigma} \sigma c_{i\sigma}^\dagger c_{i,\sigma}^{\;} = \frac{1}{2}(n_{\uparrow} - n_{\downarrow})\;,
	\end{align}
	which satisfy the usual $\mf{sl}(2,\bb C)$ commutation relations, 
	we can rewrite the Hamiltonian up to a constant as
	\begin{align}
		H = J_z + \frac{V}{n}(J_+^2 + J_-^2) + \frac{W}{n}(J_+J_- + J_-J_+)\;.
	\end{align}
	We see that $H$ is a $U(\mf{su}(2))$ Hamiltonian on a system of $n$ qubits where each qubit $i$ stores the spin of fermion $i$.
	Thus, by Thm.~\ref{thm:su2}, the LMG model can be exponentially fast-forwarded.

	\subsection{Frustration-free spin Hamiltonians at low-energies}
	\label{subsec:FrustrationFree}
	We  present a method for polynomially fast-forwarding a class of Hamiltonians denominated as frustration-free~\cite{AKLT88,PVCW08,BT09}, when the initial state is supported in a certain low-energy subspace.
	This setting can be relevant for studying quantum phase transitions, the simulation of adiabatic quantum state preparation, and more, where spectral gaps can decrease with the system size~\cite{csahinouglu2021hamiltonian}.
	For a spin system, a frustration-free Hamiltonian is $H=\sum_{X\subset \Lambda} h_X$, where each $h_X$ has the additional property $h_X \ge 0$ and the lowest eigenvalue of $H$ is zero. 
	As before, we assume $\|h_X\|\le 1$ so that $\|H\| \le L$, where $L$ denotes the number of subsets $X$ appearing in $H$. 
	While we focus on spin systems, we note that the results of this section may be applied to other systems.
	Our approach to fast-forward $H$ uses a well-known Hamiltonian simulation method based on quantum eigenvalue (or phase) estimation~\cite{Kit96,CEMM98}, but other related approaches could also be used (cf.,~\cite{low2019hamiltonian}). 
	This method first estimates (a function of) an eigenvalue of $H$ and then implements the corresponding conditional phase (cf., Ref~\cite{childs2010relationship}); see Fig.~\ref{fig:QPE}.
	To fast-forward $H$, we combine this method with the spectral gap amplification technique of Ref.~\cite{SB13} to amplify the magnitude of some eigenvalues of $H$. 
	The basic reason that allows for fast-forwarding in this case is that the amplified eigenvalues can be estimated with less accuracy -- and thus less quantum complexity -- than that needed for the estimation of the actual eigenvalues of $H$, as long as the initial state is supported in a subspace of sufficiently low energies. Our main result on fast-forwarding is given in Thm.~\ref{thm:FFfrutration-free} below, but we first prove some lemmas needed for simulating $H$ via quantum eigenvalue estimation.
	
	\begin{figure} [htbp]
		\begin{subfigure}[b]{1\textwidth}
			\centerline{\Qcircuit @C=1em @R=.7em {
					\lstick{\ket{0}_{\mc A}} & {/}\qw & \multigate{1}{U(H)} & \qw & \ustick{\ket{f(\lambda_j)}}\qw & \qw & \gate{e^{-it\lambda_j}} & \multigate{1}{U^\dag(H)} & \rstick{\ket 0_{\mc A}} \qw\\
					\lstick{\ket{\psi_j}} & {/}\qw & \ghost{U(H)} & \qw & \qw & \qw & \qw & \ghost{U^\dag(H)} & \rstick{\approx e^{-it\lambda_j}\ket{\psi_j}} \qw 
			}}
			\caption{}
		\end{subfigure}
		\\\\
		\begin{subfigure}[b]{1\textwidth}
			\centerline{\Qcircuit @C=1em @R=.7em {
					& \gate{Had} & \qw & \qw & \qw & \ctrl{4} & \multigate{3}{QFT^\dag} & \qw\\
					& \gate{Had} & \qw & \qw & \ctrl{3} & \qw & \ghost{QFT^\dag} & \qw\\
					& \gate{Had} & \qw & \ctrl{2} & \qw & \qw & \ghost{QFT^\dag} & \qw\\
					& \gate{Had} & \ctrl{1} & \qw & \qw & \qw & \ghost{QFT^\dag} & \qw\\
					& \qw & \gate{W} & \gate{W^2} & \gate{W^4} & \gate{W^8} & \qw & \qw
			}}
			\caption{}
		\end{subfigure}
		\caption{\label{fig:QPE}(a) Hamiltonian simulation via quantum eigenvalue estimation: on input state $\ket{\psi_j}\otimes \ket 0_\cA$, where $\ket{\psi_j}$ is an eigenstate of $H$ of eigenvalue $\lambda_j$, the unitary $U(H)$ outputs a state where the ancillary register $\ket{f(\lambda_j)}$ encodes information about (a function of) the eigenvalue. We then apply the phase $e^{-it\lambda_j}$ on $\ket{f(\lambda_j)}$ and invert $U(H)$. The output state is approximately $e^{-it \lambda_j} \ket{\psi_j}\otimes \ket 0_\cA$ and the overall quantum complexity of this approach depends on the approximation error. (b) An example of a building block of $U(H)$ in this case is the standard quantum phase estimation algorithm. The unitary $W=e^{-iH'/\sqrt L}$  is constructed from $H$ and $QFT$ is the quantum Fourier transform. To increase the confidence level, $U(H)$ uses this block repeatedly many times~\cite{KOS07}.
		}
	\end{figure}
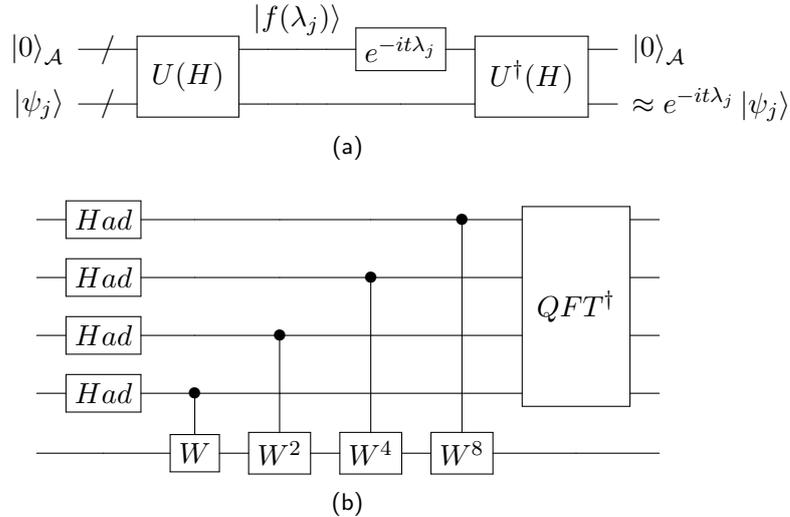
	
	\begin{lemma}[Eigenvalue amplification, from Appendix B of Ref.~\cite{SB13}]
		\label{thm:EigenvalueAmplificationPositive}
		Let $H= \sum_{X \subset \Lambda} h_X$ be frustration-free
		and let $0= \lambda_0 \leq \lambda_1 \leq \ldots$ be the (ordered) eigenvalues of $H$ associated with the eigenstates $\ket{\psi_j}$, $j=0,1,\ldots$. 
		Then, the Hamiltonian
		\begin{align}
			\label{eq:EigenvalueAmplifiedHamiltonian2}
			H' = \sum_{X \subset \Lambda} \sqrt{h_X} \otimes (\ketbras{X}{0}_\cA + \ketbras{0}{X}_\cA ) 
		\end{align}
		acts as the square root of $H$ on the subspace where the ancillary state is fixed to $\ket0_\cA$: 
		\begin{align}
			\label{eq:ClaimSqrtHamiltonian}
			(H')^2 (|\phi\rangle \otimes \ket0_\cA) = H |\phi \rangle \otimes \ket 0_\cA \;.
		\end{align}
		The states $\ket X_\cA$ are orthogonal basis states (labeled by $X$) of an ancillary system.
		Furthermore, if $\lambda_j>0$ and $\ket{\phi_j}=\sum_{X \subset \Lambda} (\sqrt{h_X}\ket{\psi_j}) \otimes \ket X_\cA/\lambda_j^{1/2}$,
		\begin{align}
			H' (\ket{\psi_j} \otimes \ket 0_\cA \pm \ket{\phi_j})=
			\pm \sqrt{\lambda_j} (\ket{\psi_j} \otimes \ket 0_\cA \pm \ket{\phi_j}) \;.
		\end{align}
		The eigenvalues of $H'$ are either zero or $\pm \lambda'_{j}=\pm \sqrt{\lambda_j}$.
	\end{lemma}
	

	\vspace{0.1cm}
	Thus, $\lambda'_{ j} \gg \lambda_j$ if $0<\lambda_j \ll 1$,
	where eigenvalue amplification occurs. Note that each $\sqrt{h_X}$ can be classically computed in constant time as this is also a $k$-local term and $k=\cO(1)$ by assumption.

	We use the method of Fig.~\ref{fig:QPE} to simulate $H$
	by estimating the eigenvalues of $H'$ via quantum eigenvalue estimation.
	If $\lambda_{j-1}<\lambda_j \ll 1$, then $|\lambda'_{ j} - \lambda'_{j-1}| \gg |\lambda_{ j} - \lambda_{j-1}|$. 
	Thus, we are able to demand a less accurate estimation of $\lambda'_{j}$'s for an accurate estimation of $\lambda_j$ in this case:

	\begin{lemma}
		\label{lem:H'AccuracyPositive}
		Let $\epsilon \ge 0$, $t>0$, and $\Delta >0$. 
		Then, an estimate of an eigenvalue $\lambda_j \le \Delta$ of $H$ within accuracy $\epsilon/(2t)$ is implied by an estimate of one of the corresponding eigenvalues $\pm \lambda'_{ j}$ of $H'$, defined  in Eq.~\eqref{eq:EigenvalueAmplifiedHamiltonian2}, within accuracy 
		\begin{align}
			\label{eq:deltadef}
			\delta= \min \left\{ \sqrt{\frac{\epsilon}{4 t}}, \frac{\epsilon}{8t \sqrt{\Delta}} \right\} \;.
		\end{align}
	\end{lemma}
	
	\begin{proof}
		The estimate of $\lambda_j$ is obtained from the squared of the estimate of $\lambda'_{j}$ or $-\lambda'_{j}$. Thus, it suffices to prove
		\begin{align}
			\label{eq:eigenest}
			\lambda_j - \epsilon/(2t) \le ( \lambda'_{j} \pm \delta )^2 \le \lambda_j + \epsilon/(2t) \; ,
		\end{align}
		or, equivalently, $2 \delta \sqrt{\lambda_j} + \delta^2 \le \epsilon/(2t)$.
		We first assume $\delta=\sqrt{\epsilon/(4t)} \le \epsilon/(8t \sqrt \Delta)$ so that
		\begin{align}
			2 \delta \sqrt{\lambda_j} + \delta^2 & \le 2 \delta \sqrt \Delta + \delta^2 \\
			& \le \epsilon/(4t) + \epsilon/(4t) \\
			& \le \epsilon/(2t).
		\end{align}
		A similar bound follows if we assume $\delta=\epsilon/(8t \sqrt \Delta) \le \sqrt{\epsilon/(4t)}$.
		Thus, it suffices to choose $\delta$ according to Eq.~\eqref{eq:deltadef} to satisfy Eq.~\eqref{eq:eigenest}.
	\end{proof}

	\begin{lemma}[Quantum eigenvalue estimation]
		\label{lem:QuantumPhaseEstimation}
		Let $H'$ be as in Eq.~\eqref{eq:EigenvalueAmplifiedHamiltonian2}. Then, there exists a quantum circuit $U(H)$ that, acting on input $\ket{\psi_j}\otimes \ket 0_\cA$, outputs an estimate of $\lambda'_{j}$
		within accuracy $\delta$ and probability $c \ge 1-\epsilon/2$ after a simple measurement.
		The quantum complexity of $U(H)$ is
		\begin{align}
			\tilde \cO \left( \log(1/\epsilon)L^2 /\delta \right) \;.
		\end{align}
		
		The $\tilde \cO$ notation hides logarithmic factors in $L$, $1/\delta$, and $\log(1/\epsilon)$.
	\end{lemma}
	
	\begin{proof}
		The result follows from previous works. In particular, we can use the high-confidence quantum phase estimation algorithm of Ref.~\cite{KOS07}. Let $W=e^{-i H'/\sqrt{L}}$
		be a unitary operator (note that $\|H'\| \le \sqrt L$). On input $\ket{\psi_j}\otimes \ket 0_\cA$, that algorithm outputs either $\lambda'_{j}/\sqrt{L}$ or $-\lambda'_{j}/\sqrt{L}$
		with probability greater or equal than $1-\epsilon/2$ and within accuracy $\delta/\sqrt{L}$ after a simple projective measurement. This is because $\ket{\psi_j}\otimes \ket 0_\cA$ is a superposition of eigenstates of $H'$ of eigenvalues $\pm\lambda'_{ j}$.
		The quantum complexity of quantum phase estimation is dominated by the number of uses of $W$, or its conditional version, which in this case is $m=\cO(\log(1/\epsilon)\sqrt{L}/\delta)$.
		
		Next, we can simulate $W$ using a known algorithm for Hamiltonian simulation. To avoid undesired overheads due to the precision factor, we consider the method of Ref.~\cite{BCC+15}; alternatively, we could use the method of Ref.~\cite{LC17}. 
		Reference~\cite{BCC+15} implements $W$ by implementing a truncated Taylor series of the exponential operator.
		That method requires a decomposition of the Hamiltonian as a linear combination of unitary operations. When $H$ is a $k$-local, $n$-spin Hamiltonian, and $k=\cO(1)$, it is simple to write $H'$ as a sum of $\cO(L)$ unitaries. 
		To this end, we encode the ancilla register of dimension $L$ into $\log L$ qubits with binary encoding. 
		This implies that the Hamiltonian $H'$ is a sum of $L$-many $(k + \log L)$-local terms.
		Each of these terms can be written as a sum of $\mathcal{O}(1)$ unitaries.
		For example, $\sqrt{h_X}$ can be processed into a sum of $\cO(1)$ $k$-local unitaries. Also,
		$\ketbras{X}{0} + \ketbras{0}{X}= \frac{1}{4}(\mathds{1} + Z_{0X}) U_{0X} (\mathds{1} + Z_{0X})$, where $Z_{0X}= \ketbra{0} + \ketbra{X} - \sum^{L}_{i =1: i \neq X} \ketbra{i}$ and $U_{0X}= \ketbras{X}{0} + \ketbras{0}{X} + \sum^{L}_{i =1: i \neq X} \ketbra{i}$ are $\log L$-local unitaries. 
		The latter can be easily expressed as products of $\cO(\log L)$ controlled bit and phase flips (two-qubit gates) using standard techniques~\cite{NC01}. 
		
		Then, we use the method of Ref.~\cite{BCC+15} to simulate $W$. This method assumes a presentation of the Hamiltonian as a linear combination of unitaries. When the evolution time is $1/\sqrt L$ and the Hamiltonian is a sum of $\cO(L)$ unitaries, and each is at most $(k+\log L)$-local, with $k=\cO(1)$, the quantum complexity of this method is $\tilde \cO (L^{3/2} \log(1/\delta'))$. The $\tilde \cO$ notation hides logarithmic factors in $L$, and $\delta'$ is the accuracy.
		
		The quantum circuit $U(H)$ is the high-confidence quantum phase estimation algorithm, which is built upon repeated many calls of Fig.~\ref{fig:QPE}(b), and where $W$ is approximated as above. The result then follows by noticing  that it suffices to choose $\delta' = \cO(\delta/m)$ for overall accuracy $\delta$. The resulting quantum complexity is
		\begin{align}
			\tilde \cO \left( \log(1/\epsilon)L^2 /\delta \right) \;,
		\end{align}
		where we dropped logarithmic factors in $L$, $1/\delta$, and $\log(1/\epsilon)$.
		
	\end{proof}

	According to Lemma~\ref{lem:QuantumPhaseEstimation}, fast-forwarding of frustration-free Hamiltonians
	can then occur when $1/\delta$ is sublinear in $t$, since the conditional phase operation of Fig.~\ref{fig:QPE}(a) has a very mild (polylogarithmic) complexity dependence on $t$.
	This can happen if $\Delta$ is sufficiently small; for example, when $\Delta$ is a decreasing function of $t$.
	The main result of this section is:

	\begin{theorem}[Polynomial fast-forwarding of frustration-free Hamiltonians at low energies]
		\label{thm:FFfrutration-free}
		Let $\{\cC_n\}_n$ denote the classes of $k$-local, $n$-spin Hamiltonians of dimension $d$ acting on $\cH_n =\left(\mathbb C^d\right)^{\otimes n}$ and $H_n=\sum_{X \subset \Lambda_n} h_X \in \cC_n$ be frustration-free as above, with $\|h_X\| \le 1$.
		The number of subsets $X$ in $H_n$ is $L(n)=\poly(n)$.
		Let $\epsilon > 0$, $0< T(n) \le d^{2n}$, and $\Delta(n) \le 1/T(n)$.
		Consider the subspaces $\{\cS_n \subseteq \cH_n\}_n$
		where $\cS_n$ is spanned by all the eigenstates of $H_n$ of eigenvalues in $[0,\Delta(n)]$.
		Then, there exist quantum circuits $\{V_n(t)\}_{n,t}$ in the standard gate model that are $(T(n),\epsilon,G(n))$-fast-forwarding the Hamiltonians $\{H_n\}_n$ on subspaces $\{\cS_n\}_n$,
		where $G(n)=\tilde \cO(L(n)^2 \sqrt{T(n)})$. 
		In particular,
		if $l(n)=\Omega(n T(n))$, we obtain $G(n)/l(n)=\tilde\cO(L^2(n) /(n\sqrt{T(n)}))$, which decreases asymptotically with $n$ if $T(n)= \tilde\Omega(L^4(n)/n^2)$.
	\end{theorem}
	
	\begin{proof}
		The definition of fast-forwarding considers the simulation of $H_n$ for all times $0 \le t \le T(n)$ and it suffices to consider the maximum $T(n)$ to place an upper bound on the quantum complexity $G(n)$. 
		The simulation method we use is the one described in Fig.~\ref{fig:QPE}, where we first estimate the eigenvalues of $H'_n$ to implement the conditional phase
		that approximates $e^{-i t \lambda_j}$.

		Let $\delta_n=\epsilon/(8 \sqrt{T(n)}) < \sqrt{\epsilon/(4T(n)})$ and $c=1-\epsilon/4$.
		Using the high-confidence version of quantum phase estimation of Ref.~\cite{KOS07} and building $U(H)$ according to Lemma~\ref{lem:QuantumPhaseEstimation}, the output state
		of $U(H)$ is $\ket{\psi_j} \otimes \ket{f(\lambda_j)}_\cA$
		when the input state is $\ket{\psi_j}\otimes \ket0_\cA$. Here, $\ket{f(\lambda_j)}_\cA$ is a linear combination of basis states and has the following properties: with probability at least $1-\epsilon/4$, it is supported on a subspace where the corresponding registers encode the eigenvalue $\lambda'_j$ of $H'$ within precision $\delta_n$. Equivalently, with probability at least $1-\epsilon/4$, it is supported on a subspace where the corresponding registers encode the eigenvalue $\lambda_j$ of $H$ within precision $\epsilon/(2T(n))$; see Lemma~\ref{lem:H'AccuracyPositive}.
		The conditional phase operation then introduces an error that is at most $t \epsilon/(2T(n))\le \epsilon/2$. Inverting $U(H)$ can also add an additional factor of $\epsilon/4$ to the error. Thus, this approach simulates $H_n$ for time $t \le T(n)$ with overall error bounded by $\epsilon/4 + \epsilon/2 + \epsilon/4 = \epsilon$. The resulting circuits following this approach are the $\{V_n(t)\}_{n,t}$, given in Fig.~\ref{fig:QPE}.
		
		The quantum complexity for this method is dominated by that of $U(H)$, analyzed in Lemma~\ref{lem:QuantumPhaseEstimation}. For constant error, this is $\tilde \cO((L^2(n)) \sqrt{T(n)})$, where we dropped logarithmic factors in $L(n)$ and $T(n)$. This complexity is sublinear in $T(n)$.
		Assuming $l(n)=\Omega(n T(n))$, then $G(n)/l(n)\rightarrow 0$ for cases where, for example, $T(n)$ is exponential in $n$.
	\end{proof}
	
	The dominant factor in the scaling of $G(n)$ in Thm.~\ref{thm:FFfrutration-free} is coming from $\sqrt{T(n)}$, and we refer to this case as ``quadratic'' fast-forwarding.
	While Thm.~\ref{thm:FFfrutration-free} focuses on the case where $\Delta$ decreases as $1/T$, it is clear that if $\Delta$ decreases as $1/T^\alpha$, for any positive $\alpha \ne 1$, other types of polynomial fast-forwarding can result. 
	
	As mentioned in Sec.~\ref{sec:def}, while $l(n)$ may not be known, several results in the literature strongly suggest that $l(n)=\Omega(nT(n))$ for $T(n) \le d^{2n}$.

	\section{Fast-forwarding of fermionic and bosonic systems}
	\label{sec:Fermionic/BosonicModels}
	Quantum systems obeying various particle statistics such as fermions or bosons play an important role in physics, including condensed matter and quantum field theories, quantum chemistry, and more. 
	In second quantization, Hamiltonians of $n$-mode fermionic and bosonic systems are written in terms of annihilation and creation  operators $c_i^{\;},c_i^\dag$, respectively, where $i=1,\ldots,n$.
	Fermionic operators satisfy the canonical anticommutation relations
	\begin{align}
		\{c_i^{\;},c_j^\dagger\} & = \delta_{ij} \; , \\
		\{c_i^{\;},c_j^{\;}\} & = 0 \;,  
	\end{align}
	while bosonic operators satisfy the canonical commutation relations
	\begin{align}
		[c_i^{\;},c_j^\dagger]& =\delta_{ij} \; , \\
		[c_i^{\;},c_j^{\;}]& =0 \; .
	\end{align}
	A fermionic or bosonic Hamiltonian is written as
	\begin{align}
		H & = \sum_{I\subseteq \{1,\dots,n\}} h_I \; ,
	\end{align}
	where each $h_I$ describes interactions among a subset $I$ of fermionic or bosonic modes. More precisely, $h_I$ is a sum of terms of the form $a (c_{i_1}^\dag)^{e_1}\dots(c_{i_k}^\dag)^{e_k}(c_{j_1}^{\;})^{f_1}\dots (c_{j_l}^{\;})^{f_l}$, and $a \in \mathbb C$. Due to the exclusion principle, $e_i=f_j=1$ for fermionic systems whereas $e_i \ge 1$ and $f_i \ge 1$ for bosonic systems, for all $i \le k$ and $j \le l$. The Hilbert space $\cH$ is the standard Fock space with the proper symmetrization~\cite{reed1978methods}.
	
	As for spin systems, we 
	can define the degree of $H$ to be the largest sum of the exponents $e_1+\dots+e_k+f_1+\dots+f_l$ in the terms of $h_I$
	and the {\em weight} of $H$ to be the largest value of $k+l$. We will be particularly interested in those physically relevant cases where the degree and weight of $H$ are constant and further assume $|a|\le 1$.
	In this context, a fermionic or bosonic Hamiltonian sequence $\{H_n\}_n$ is meaningful if we can construct $H_n$ from $H_{n-1}$ by adding an interaction term $v_n$ that acts non-trivially on the $n$-th mode. Our definition of fast-forwarding also depends on the model of quantum computation, and in this case, the fermionic and bosonic models are such that
	the elementary quantum gates have the form $e^{i \theta_l O_l}$, where $O_l$ is a fermionic or bosonic operator of bounded weight and degree and $|\theta_l| \le 1$; see Def.~\ref{def:modelsofQC}. Note that these models are different than the standard gate model, but a number of mappings can be used to simulate fermionic or bosonic systems with qubit quantum computers if desired~\cite{SOGKL02,SOKG03}.

	Very few results in the literature can be used to address the no-fast-forwarding line for fermionic or bosonic systems.
	Such are the known results for local spin systems~\cite{HHK+18}
	that, under some assumptions, could also be applied to this setting. 
	(The resulting complexity bounds will be weakened by overheads due to mappings between various models of quantum computing.) For certain classes of fermionic or bosonic systems, these results imply a quantum complexity that is at least linear in $t$, in certain range.
	In general, we conjecture $l(n)= \Omega({\rm poly}(n)T(n))$ for the fermionic and bosonic systems discussed in the following sections, where $T(n)=\exp(\Omega(n))$.  We will prove
	exponential fast-forwarding of quadratic fermionic and certain quadratic bosonic Hamiltonians under this assumption, where the weight of $H$ is at most 2.
	We do this by using a Lie-algebra diagonalization approach that can be applied more generally.

	\subsection{Lie-algebra diagonalization}
	\label{subsec:LieAlgebraDiagonalization}
	We describe an approach for fast-forwarding Hamiltonians that are elements of a Lie algebra of {\em small} dimension.
	Hamiltonian simulation via diagonalization was also recently explored in Refs.~\cite{cirstoiu2020variational,KSFDK21} with a focus on near term simulation of qubit Hamiltonians.
	Our method, which works for fermionic, bosonic, and other types of Hamiltonians, uses a Lie algebraic version of the Jacobi eigenvalue algorithm for diagonalizing matrices~\cite{Wildberger93,Som05,Som19}, which we review. 
	
	Suppose the Hamiltonian $H$ is an element of a real compact semisimple Lie algebra $\mathfrak g$ with Lie group $G$. The Cartan-Weyl basis decomposition of the complexified Lie algebra implies
	\begin{align}
		\mathfrak g=\mathfrak h \oplus \left(\bigoplus_{\alpha\in R_+} \mathfrak g_\alpha\right) \;,
	\end{align}
	where $\mf h$ is the Cartan subalgebra, $R_+\subseteq \mf h^*$ are the positive roots and each $\mf g_\alpha$ is the associated two-dimensional space corresponding to the roots $\alpha$ and $-\alpha$. Because $\mf g$ is compact, each Lie subalgebra generated by $\mathfrak g_\alpha$ is isomorphic to $\mathfrak{su}(2)$ and spanned by generalized Pauli elements $X_\alpha,Y_\alpha\in \mathfrak g_\alpha, Z_\alpha\in\mathfrak h$ with commutation relations
	\begin{align}
		[X_\alpha,Y_\alpha] &= 2iZ_\alpha \; , \\
		[Y_\alpha,Z_\alpha] &= 2iX_\alpha \; ,\\
		[Z_\alpha,X_\alpha] &= 2iY_\alpha \; .
	\end{align}
	The Lie-algebra diagonalization approach
	is simply a sequence of $\mathfrak{su}(2)$ rotations that eliminate the off-diagonal elements on $2 \times 2$ blocks gradually.
	In the first step we set $H^{(0)} = H$.  At each step $j \ge 0$, we write
	\begin{align}
		H^{(j)} = H_D^{(j)} + \sum_{\alpha\in R_+} (a_\alpha^{(j)} X_\alpha + b_\alpha^{(j)} Y_\alpha).
	\end{align}
	The $H_D^{(j)} \in \mathfrak h$ term is the diagonal component of the Hamiltonian $H^{(j)}$. The off-diagonal coefficients
	are $a_\alpha^{(j)}$, $b_\alpha^{(j)} \in \mathbb R$ and the goal is to make these small. To this end, we eliminate the  $a_{\alpha_j}^{(j)},b_{\alpha_j}^{(j)}$'s of largest norm  at each step. This is accomplished by performing the transformation
	\begin{align}
		H^{(j)}\rightarrow H^{(j+1)}= (V^{(j+1)})^\dagger H^{(j)} V^{(j+1)} \;.
	\end{align}
	The unitaries $V^{(j)}$ are  $\mathfrak{su}(2)$ rotations given by
	\begin{align}
		V^{(j+1)} = e^{i(\pi_x X_{\alpha_j}+\pi_y Y_{\alpha_j})} \;,
	\end{align}
	where the coefficients $|\pi_{x,y}| \le \pi$ can be determined from $H^{(j)}$ in time polynomial in $\dim \mathfrak g$;
	more details can be found in Ref.~\cite{Som19}.

	Convergence of the sequence $\{H^{(j)}\}_j$ to a diagonal operator can be measured using the distance from each Hamiltonian to the Cartan subalgebra, defined by the norm of the off-diagonal terms
	\[d_{\mf h}(H^{(j)}) = \left(8\sum_{\alpha\in R_+}\left(|a_\alpha^{(j)}|^2+|b_\alpha^{(j)}|^2\right)\right)^{1/2}.\]
	Wildberger proved that this algorithm converges exponentially~\cite{Wildberger93}:
	
	\begin{lemma}[Wildberger]
		\label{lem:wildberger}
		Let $H$ be an element of a real compact semisimple Lie algebra $\mf g$ with Lie group $G$. Let $l=|R_+|<\frac{1}{2}\dim\mf g$. Then, there is a sequence of $r$ unitaries $V^{(j)}\in \exp(i\mf g_{\alpha_j})\subseteq G$ such that
		\begin{align}
			d_{\mf h}(H^{(r)})\le \left(\frac{l-1}{l}\right)^{r/2}d_{\mf h}(H) \;.
		\end{align}
	\end{lemma}

	The Hamiltonian simulation approach via Lie-algebra diagonalization is given in Fig~\ref{fig:diagonalization}. In principle, each unitary $V^{(j)}$ is a sequence of elementary gates in a corresponding model of quantum computing determined by $\mathfrak g$. For simplicity, we define this model so that each $V^{(j)}$ is a single elementary gate of unit cost. 
	Additionally, the value of $r$ can be chosen so that the desired accuracy in the diagonalization is achieved. In particular, there exists
	\begin{align}
		\label{eq:rbound}
		r=\cO((\dim \mathfrak g) \log(1/\varepsilon))
	\end{align}
	that achieves $d_{\mf h}(H^{(r)})\le \varepsilon d_{\mf h}(H)$.
	The diagonal operation $e^{-i t H_D^{(r)}}$
	can be implemented with quantum complexity that depends on $t$ and the approximation error.
	We obtain:
	
	\begin{lemma}
		\label{lem:diagonalsimulation}
		Let $H$ and $H^{(j)} \in \mathfrak g$ be as above, $\epsilon>0$, and $T>0$. Let $G_D(T,\epsilon)$ be an upper bound on the quantum complexity of approximating $e^{-i t H_D^{(r)}}$ within error $\epsilon/2$ for all $t \le T$ and assume $\|H^{(r)} - H_D^{(r)}\| \le \epsilon/(2T)$. Then, there exist quantum circuits $\{V(t)\}_t$ that simulate $H$ with parameters $(T,\epsilon,2r+G_D(T,\epsilon))$.
	\end{lemma}
	\begin{proof}
		We use the circuit of Fig.~\ref{fig:diagonalization} to implement $V(t)$. Since each $V^{(j)}$ and $(V^{(j)})^\dagger$ has unit cost according to the assumptions, the quantum complexity of the circuit in Fig.~\ref{fig:diagonalization} is $2r$ plus the cost of implementing $e^{-it H_D^{(r)}}$. The correctness of this simulation method follows from simple norm properties ($t \le T$):
		\begin{align}
			\| e^{-i t H} - V^{(1)} \ldots V^{(r )}e^{-i t H_D^{(r)}} (V^{(r )})^\dagger \ldots (V^{(1)})^\dagger \| & = \| e^{-i t H^{(r)}} - e^{-i t H_D^{(r)}} \|  \\
			& \le \|t (H^{(r)} -H_D^{(r)} )\| \\
			& \le \epsilon/2 \;.
		\end{align}
		If, in addition, we approximate $e^{-it H_D^{(r)}}$ with a quantum circuit $U$ of quantum complexity at most $G_D(T,\epsilon)$ such that $\|e^{-it H_D^{(r)}}-U \|\le \epsilon/2$, then the overall approximation error is bounded by $\epsilon$ and the overall quantum complexity is $2r+G_D(T,\epsilon)$.
	\end{proof}
	
	\begin{figure} [htbp]
		\centerline{\Qcircuit @C=1em @R=.7em {
				& {/}\qw & \gate{(V^{(1)})^\dag} & \gate{(V^{(2)})^\dag} & \qw & \dots &  & \gate{(V^{(r)})^\dag} & \gate{ e^{-itH_D^{(r)}}} & \gate{ \phantom{(}V^{(r)}\phantom{)}^{\;}} & \qw & \dots & & \gate{\phantom{(}V^{(2)}\phantom{)}^{\;}} & \gate{\phantom{(}V^{(1)}\phantom{)}^{\;}} & \qw
		}}
		\caption{\label{fig:diagonalization}Hamiltonian simulation via Lie-algebra diagonalization. The quantum circuit describes the unitary $V(t)$ that simulates $U(t)=e^{-itH}$.
		}
	\end{figure}
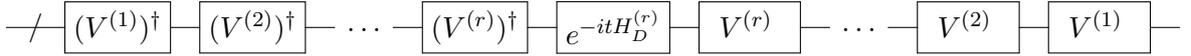

	Of particular interest for fast-forwarding via Lie-algebra diagonalization is the case where $l$ is small relative to the dimension of the Hilbert space. This occurs for many quantum systems, such as quadratic fermionic systems, where the dimension of $\mathfrak g$ is only polynomial in $n$. In this case, the Lie-algebra diagonalization method described above converges quickly, and to satisfy $\|H^{(r)}- H_D^{(r)}\|\le \epsilon/(2T)$, it suffices to choose $r$ that is polynomial in $n$ but only polylogarithmic in $T/\epsilon$. If, in addition, the quantum complexity of approximating $e^{-i t H_D^{(r)}}$ within error $\epsilon$ is also polynomial in $n$ and polylogarithmic in $T/\epsilon$, for $t \le T$, we attain exponential fast-forwarding:

	\begin{theorem}[Exponential fast-forwarding via Lie-algebra diagonalization]
		\label{thm:diagonalFF}
		Let $\{\cC_n\}_n$ denote classes of Hamiltonians acting on spaces $\{\cH_n\}_n$ and $H_n \in \cC_n$, $H_n \in \mathfrak g_n$, where  $\mathfrak g_n$ is a compact, semisimple Lie algebra of $\dim \mathfrak g_n ={\rm poly}(n)$. Assume 
		$d_{\mf h}(H_n)={\rm poly}(n)$ and that, for any $X \in \mf g_n$,
		\begin{align}
			\label{eq:distancebounds}
			\|X-X_D\|\le p(n)d_{\mf h}(X)\;, \quad d_{\mf h}(X)\le q(n)\|X-X_D\| \;,
		\end{align}
		where $p(n),q(n)=\poly(n)$ and $X_D$ is the (diagonal) projection of $X$ onto $\mathfrak h$. Let $\epsilon > 0$ and $T(n) = \tau^n$ for some $\tau>1$. 
		If the quantum complexity of simulating the diagonal version of $H_n$ for time $t \le T(n)$ within precision $\epsilon/2$ is $G_D(t,\epsilon)=\poly(n,\log(t \|H_n\|/\epsilon))$, then there exist quantum circuits $\{V_n(t)\}_{n,t}$ that are $(T(n),\epsilon,G(n))$-fast-forwarding the Hamiltonians $\{H_n\}_n$ on the spaces $\{\cH_n\}_n$, where $G(n)={\rm poly}(n)$. In particular, if $l(n)=\Omega(n T(n))$, we obtain exponential fast-forwarding.
	\end{theorem}
	
	\begin{proof}
		The proof is a direct consequence of Lemma~\ref{lem:diagonalsimulation} when considering the sequence of Hamiltonians $\{H_n\}_n$ and under the assumptions on the asymptotic scalings of quantum complexities. In particular, the bounds in Eqs.~\eqref{eq:distancebounds} and $\dim \mf g_n = \poly(n)$ allow us to set $r=\poly(n)$. The classical complexity of finding the unitaries $V^{(j)}$ is polynomial in the dimension of the Lie algebra and polylogarithmic in the precision, and thus is ${\rm poly}(n)$ under all the assumptions.
	\end{proof}
	
	The requirements to achieve exponential fast-forwarding using the Lie-algebra diagonalization approach are satisfied by some important classes of Hamiltonians, such as quadratic fermionic and bosonic Hamiltonians, as we show below.
	While this diagonalization approach is efficient, 
	in some cases  it can be less efficient than directly decomposing the time evolution operator into a sequence of elementary quantum gates using a different approach~\cite{Kivlichan18,Jiang18}.

	\subsubsection{Quadratic fermionic Hamiltonians}
	\label{subsec:QuadraticFermionicHamiltonians}
	As one example of exponential fast-forwarding via Lie-algebra diagonalization, we consider fermionic Hamiltonians which are quadratic in the creation and annihilation operators. 
	These are written as
	\begin{align}
		\label{eq:fermionicH}
		H = \sum_{i,j} \alpha_{ij}c_i^\dagger c_j^{\;} + \beta_{ij}c_i^{\;}c_j^{\;} - \beta_{ij}^* c_i^\dagger c_j^\dagger \; ,
	\end{align}
	where we assume $\alpha_{ij}=(\alpha_{ji})^*$, $|\alpha_{ij}|\le 1$, and 
	$|\beta_{ij}|\le 1$. Up to an additive constant, the terms in $H$ span a compact Lie algebra $\mathfrak g$ isomorphic to $\mf{so}(2n,\bb R)$, where the Cartan subalgebra is spanned by $\{c^\dagger_i c^{\;}_i -1/2\}_i$~\cite{GilFen83}. This satisfies
	$\dim \mathfrak g= n(2n-1)$, being polynomial in $n$, and we can obtain a bound of the form \eqref{eq:distancebounds}. Thus we can use the analysis of the Sec.~\ref{subsec:LieAlgebraDiagonalization} to demonstrate exponential fast-forwarding. We obtain:

	\begin{theorem}[Exponential fast-forwarding of quadratic fermionic Hamiltonians]
		\label{thm:fermionicFF}
		Let $\{\cC_n\}_n$ denote the classes of $n$-mode fermionic Hamiltonians acting on Fock spaces $\cH_n$ and
		$\{H_n \in \mf{so}(2n,\bb R)\}_n$ be as in Eq.~\eqref{eq:fermionicH}. Let $\epsilon > 0$ and $T(n)=\tau^n$ for some $\tau>1$. 
		Then, there exist quantum circuits $\{V_n(t)\}_{n,t}$ that are $(T(n),\epsilon,G(n))$-fast-forwarding the Hamiltonians $\{H_n\}_n$ on the spaces $\{\cH_n\}_n$, where $G(n)=\cO(n^2\log (T(n)))$, and thus $G(n)={\rm poly}(n)$. In particular, if $l(n)=\Omega(n T(n))$, we obtain exponential fast-forwarding.
	\end{theorem}
	
	\begin{proof}
		The circuits $\{V_n(t)\}_{n,t}$ to simulate the Hamiltonians $\{H_n\}_n$
		are given in Fig.~\ref{fig:diagonalization}. For precision $\epsilon$, it suffices to choose $r$ in
		Lemma~\ref{lem:wildberger} such that $\|H_n^{(r)}- H_{n,D}^{(r)} \| =\cO(\epsilon/T(n))$, where $H_n^{(r)}$ is the Hamiltonian at the $r$-th step of the Lie-algebra diagonalization procedure that starts with $H_n$, and $H_{n,D}^{(r)}$ is its projection onto the Cartan subalgebra.
		The elements of the fermionic $\mf{so}(2n,\bb R)$ algebra satisfy the bounds of Eq.~\eqref{eq:distancebounds}. A precise analysis allows us to take $p(n)=\frac 1 2\sqrt{\frac{n}{n-1}}$ in Eq.~\eqref{eq:distancebounds}. Since $|\alpha_{ij}|,|\beta_{ij}|\le 1$, the initial distance to the Cartan subalgebra is $d_{\mf h}(H_n)=\cO(n^2)$. Thus, we choose
		$\varepsilon=\epsilon/(T(n)p(n)n^2)$ in Eq.~\eqref{eq:rbound},
		and it suffices to run
		\begin{align}
			r&=\cO\left(\dim\mf{so}(2n,\bb R)\log\left(\frac{T(n)p(n)n^2}{\epsilon}\right)\right) \\
			& = \cO\left(n^2\log (T(n) )\right)
		\end{align}
		iterations of the Lie-algebra diagonalization method for each $n$, which determines the complexity of $V_n(t)$. Each $\mf{su}(2)$ rotation $V^{(j)}$ in Fig.~\ref{fig:diagonalization} is the exponential of a weight two fermionic operator and is thus an elementary gate in the fermionic model of computation. Since the exponential of the diagonal Hamiltonian $H_{n,D}^{(r)}$ can be implemented as a sequence of  $\cO(n)$ rotations $e^{-i\theta_jc_j^\dag c^{\;}_j}$, $|\theta_j | \le \pi$, we obtain the desired total quantum complexity $G(n)=\cO(n^2\log (T(n)))$.
	\end{proof}

	\subsubsection{Quadratic bosonic Hamiltonians}
	\label{subsec:QuadraticBosonicHamiltonians}
	The approach of Sec.~\ref{subsec:LieAlgebraDiagonalization} may also be applied to quadratic bosonic Hamiltonians. However, for bosonic Hamiltonians containing non-number-conserving terms $c_ic_j, c_i^\dag c_j^\dag$, the real Lie algebra spanned by the Hermitian terms is noncompact (with complexification $\mf{sp}(2n,\bb C)$), and the Lie-algebra diagonalization technique cannot directly be used. Therefore, we consider number-conserving quadratic bosonic Hamiltonians
	\begin{align}
		\label{eq:bosonicH}
		H = \sum_{i,j} \alpha_{ij}c_i^\dagger c_j^{\;} \;,
	\end{align}
	where $|\alpha_{ij}|\le 1$ and $\alpha_{ji}=\alpha_{ij}^*$. This simplification also allows us to restrict to subspaces $\mc S_m$ of $m$ total bosons, on which $H$ is a bounded operator. 
	It is clear that $\mc S_m$ is invariant under $H$. 
	By adding a constant to $H$, we identify the real Lie algebra $\mf g$ generated by the Hermitian terms of $H$ as $\mf{su}(n)$, which has polynomial dimension $\dim \mf g=n^2-1$.
	The Cartan subalgebra consists of all elements of the form $\sum_j \gamma_j c_j^\dag c_j^{\;}$ where $\sum_j\gamma_j=0$ and $\gamma_j \in \mathbb{R}$.
	Then we can use the analysis of Sec.~\ref{subsec:LieAlgebraDiagonalization} to show exponential fast-forwarding.
	
	\begin{theorem}[Exponential fast-forwarding of number-conserving quadratic bosonic Hamiltonians]
		\label{thm:bosonicFF}
		Let $\{\cC_n\}_n$ denote the classes of $n$-mode, number-conserving bosonic Hamiltonians acting on Fock spaces $\cH_n$ and
		$\{H_n \in \mf{su}(n)\}_n$ be as in Eq.~\eqref{eq:bosonicH}. Let $\epsilon > 0$ and $T(n)=\tau^n$ for some $\tau>1$. Consider the subspaces $\mc S_{m(n)}$ with $m(n)=\poly(n)$.
		Then, there exist quantum circuits $\{V_n(t)\}_{n,t}$ that are $(T(n),\epsilon,G(n))$-fast-forwarding the Hamiltonians $\{H_n\}_n$ on the subspaces $\{\cS_{m(n)}\subseteq \cH_n\}_n$, where $G(n)=\cO(n^2\log (T(n)))$, and thus $G(n)={\rm poly}(n)$. In particular, if $l(n)=\Omega(n T(n))$, we obtain exponential fast-forwarding.
	\end{theorem}
	
	\begin{proof}
		The proof is similar to that of Thm.~\ref{thm:fermionicFF}. By restricting to the subspace $\mc S_{m(n)}$, we may obtain bounds of the form \eqref{eq:distancebounds}, where $p(n)=\frac{m(n)}{\sqrt{2n}}$. Since the initial distance to the Cartan subalgebra is $d_{\mf h}(H_n)=\cO(n^2)$, we apply the Lie-algebra diagonalization method for
		\begin{align}
			r &= \cO\left(\dim\mf{su}(n)\log \left(\frac{T(n)p(n)n^2}{\epsilon}\right)\right)\\
			&= \cO\left(n^2\log (T(n))\right)
		\end{align}
		iterations. As in the fermionic case, the $\mf{su}(2)$ rotations $V^{(j)}$ are elementary gates and the diagonal unitary $e^{-itH_{n,D}^{(r)}}$ can be decomposed as a sequence of $\cO(n)$ rotations $e^{-i\theta_j c_j^\dag c^{\;}_j}$, $|\theta_j|\le \pi$. Therefore, we obtain a total of $G(n)=\cO(n^2\log (T(n)))$ gates.
	\end{proof}

	\section{Fast-forwarding and energy measurements}
	\label{sec:energy}
	We revisit and further develop a connection between fast-forwarding and the time-energy uncertainty principle discussed in Ref.~\cite{AA17}, using the setting described in Sec.~\ref{sec:def}.
	Our result, which also holds for polynomial fast-forwarding and polynomially-precise energy measurements, is roughly as follows.
	Suppose we can simulate a (normalized) Hamiltonian for time $T$ using $G(T)$ elementary gates.
	Then we can measure the eigenenergies with accuracy $ \cO(1/ T)$ and high confidence using $\tilde \cO(G(T))$ elementary gates, where we dropped polylogarithmic factors in $T$.
	Conversely, suppose we can measure the eigenenergies of a Hamiltonian with accuracy $\delta E$ and high confidence using $G(\delta E)$ elementary gates.
	Then, we can simulate the Hamiltonian for time $\cO(1/{\delta E})$ using $\cO(G(\delta E))$ elementary gates.

	To prove this result, we need to provide precise definitions of energy measurements that fit with those in Sec.~\ref{sec:def}.
	
	\begin{definition}[Accuracy of energy measurements]
		An energy measurement has accuracy $\delta E$ and confidence level $\eta >0$ if, on any input eigenstate of a Hamiltonian, we measure an outcome $E'$ that satisfies
		\begin{align}
			\label{eq:accuracyenergy}
			\Pr\left(|E'-E|\le \delta E\right) \ge \eta \; ,
		\end{align}
		where $E$ is the energy of the eigenstate.
	\end{definition}
	
	\begin{definition}[Parameters of energy measurements]
		\label{def:SEEM}
		A sequence of Hamiltonians $\{H_n\}_n$ acting on spaces $\{\cH_n\}_n$ has energy measurement parameters $(\eta,\delta E(n),\xi,G(n))$ if there exist unitary operations $\{U_n\}_n$ and quantum circuits $\{V_n\}_n$ that satisfy:
		\begin{enumerate}
			\item For any $n$ and any eigenstate $\ket{\psi_E}$ of $H_n$ with energy $E$, the unitary $U_n$ acts as
			\begin{align}
				U_n\ket{\psi_E} \otimes \ket 0_\cA = \ket{\psi_E} \otimes \sum_{E'}\alpha_{E'} \ket{E'}_\cA \;,
			\end{align}
			and
			\begin{align}
				\sum_{E': |E-E'|\le\delta E(n)} |\alpha_{E'}|^2 \ge \eta \;.
			\end{align}
			Therefore, a projective measurement of the ancillary system outputs an estimate $E'$ that satisfies Eq.~\eqref{eq:accuracyenergy}.
			\item For any $n$, the quantum circuit $V_n$ can be implemented with at most $G(n)$ elementary gates and satisfies
			\begin{align}
				\left\|(U_n-V_n)\ket\psi \otimes \ket 0_\cA \right\|\le\xi \;,
			\end{align}
			for all $\ket\psi \in \cH_n$.
		\end{enumerate}
		\sloppy Similarly, 
		the sequence of Hamiltonians $\{H_n\}_n$ has energy measurement parameters $(\eta,\delta E(n),\xi,G(n))$ on invariant subspaces $\{\cS_n\subseteq \cH_n\}_n$ of the Hamiltonian if the two conditions hold for any eigenstate $\ket{\psi_E}\in \cS_n$ and all $\ket\psi\in \cS_n$, respectively.
	\end{definition}
	
	In the first condition, the ancillary states $\ket{E'}_\cA$ encode the energy in, for example, binary form. 
	Note that these definitions are essentially the same as in Ref.~\cite{AA17} with the addition of $G(n)$ for the gate count, which is necessary in our setting because we are considering a more general form of efficient energy measurements.
	
	Now we state the equivalence between fast-forwarding and energy measurements precisely. For simplicity, we restrict to qubit systems, although generalizations can be made to other models of computation.
	
	\begin{theorem}[Fast-forwarding and precise energy measurements]
		\label{thm:FF&SEEM}
		Let $\{H_n\}_n$ be a sequence of (qubit) Hamiltonians with norms $\|H_n\|\le 1$. Then the following statements hold:
		\begin{enumerate}
			\item If $\{H_n\}_n$ can be simulated with parameters $(T(n),\epsilon,G(n))$, then for any constant $\eta<1$, $\{H_n\}_n$ has energy measurement parameters $(\eta,\cO(\frac{1}{T(n)}),\cO(\epsilon\log T(n)),\cO((\log T(n))^2+G(n)\log T(n)))$.
			\item Conversely, if $\{H_n\}_n$ has energy measurement parameters $(\eta,\delta E(n),\xi,G(n))$, then for any constant $\alpha>0$, $\{H_n\}_n$ can be simulated with parameters $(\cO(\frac{1}{\delta E(n)}),\alpha\eta+2(1-\eta+\xi),\cO(G(n)))$.
		\end{enumerate}
		The same statements hold on subspaces $\mc S_n\subseteq \left(\bb C^2\right)^{\otimes n}$ if $\left\|\left.H_n\right|_{\mc S_n}\right\|\le 1$.
	\end{theorem}
	\begin{proof}
		The proof of the theorem uses the same technique as in Ref.~\cite{AA17} and is reproduced in Appendix~\ref{appx:energy}. To prove the first statement, we use quantum phase estimation to measure the energy. This approach requires implementing the unitaries $e^{-i t H_n}$
		controlled on an ancillary register of $\cO(\log(T(n))$ qubits and for times $t \le T(n)$. To prove the second statement, we measure the energy $E$ of an eigenstate $\ket{\psi_E}$ on an ancillary register and implement the phase $e^{-itE}$ on that register using standard techniques.
	\end{proof}
	
	According to Thm.~\ref{thm:FF&SEEM}, we see that various types of fast-forwarding (e.g., polynomial) result in various types of precise energy measurements that are not necessarily exponential.

	
	\section{Conclusions and outlook}
	\label{sec:ConclusionsAndOutlook}

	We studied the problem of fast-forwarding quantum evolution in various physical settings and under fairly general conditions, going beyond previous studies~\cite{AA17}. We provided a definition of fast-forwarding that considers the model of quantum computation, the classes of Hamiltonians, and properties of the initial states, and used it to demonstrate exponential and polynomial fast-forwarding in quantum systems of different particle statistics. These include spin systems that have a permutation-invariant property, frustration-free spin systems, and quadratic fermionic and bosonic systems. Our techniques could be used to demonstrate fast-forwarding of Hamiltonians not discussed in our work.
	For example, although we focused on quadratic fermionic and bosonic Hamiltonians as an application of Lie-algebra diagonalization in Sec.~\ref{subsec:LieAlgebraDiagonalization}, that method can be used for any Hamiltonian that is an element of a Lie algebra of small dimension;
	a recent result in Ref.~\cite{KSFDK21} considers a related problem.
	
	Many open questions remain. First, we do not provide a characterization of all Hamiltonians that can be fast-forwarded but rather some examples. It would be interesting to understand this characterization better, as it has important consequences in quantum simulation and quantum complexity. Second, some examples of exponential fast-forwarding, such as those obtained via Lie-algebra diagonalization, are related to Hamiltonians that admit an efficient classical solution, for which spectral properties and related quantities may be computed classically efficiently~\cite{baxter2016exactly, korepin1997quantum, sutherland2004beautiful,somma2006efficient, murg2012algebraic}.
	Developing the connection between exponential fast-forwarding and efficient solvability further will be important for understanding the advantages of exponential fast-forwarding. Third, the polynomial fast-forwarding of Sec.~\ref{subsec:FrustrationFree} applies to frustration-free spin Hamiltonians, but it might also be applied to the simulation of more general quantum field theories or other condensed matter systems, as we are often interested in the low-energy dynamics of such Hamiltonians. 
	Fourth, our assumptions on the no-fast-forwarding line
	are based on no-fast-forwarding theorems in the literature, and extending these theorems to more general quantum systems would allow us to understand the level of fast-forwarding better when using our simulation methods.
	We expect that these and other questions will form the basis of further studies.

	
	\section{Acknowledgements}
	\label{sec:acknowledgements}
	
	We acknowledge support from the U.S. Department of Energy, Office of Science, High-Energy Physics and Office of Advanced Scientific Computing Research, under the Accelerated Research in Quantum Computing (ARQC) program. 
	This material is also based upon work supported by the U.S. Department of Energy, Office of Science, National Quantum Information Science Research Centers, Quantum Science Center.
	Los Alamos National Laboratory is managed by Triad National Security, LLC, for the National Nuclear Security Administration of the U.S. Department of Energy under Contract No. 89233218CNA000001.

	\newpage
	
	\bibliographystyle{unsrtnat}
	\bibliography{FastForward}

	\newpage
	
	\appendix
	
	\section{Complexity of block diagonalizing permutation-invariant qubit Hamiltonians}
	\label{appx:schur}
	
	We give a circuit description of the Schur transform $U_{\rm Sch}$ for qubit (spin-$1/2$) systems, following Ref.~\cite{BCH06}. Although it is not the most efficient implementation in terms of gate complexity, it has the advantage of being explicit (without requiring a separate ancillary register for classical computation), and the unary encoding for the $\ket J$ and $\ket q$ registers is suitable for simulation of permutation-invariant qubit Hamiltonians in Thm.~\ref{thm:su2}.
	
	\begin{lemma}\label{lem:schur}
		The Schur transform $U_{\rm Sch}$ for $n$-qubit systems can be implemented with $\cO(n^3)$ two-qubit gates.
	\end{lemma}
	
	\begin{proof}
		As the representation space of $U(\mf{su}(2))$, the Hilbert space $\mc H=\left(\bb C^2\right)^{\otimes n} = \mc Q_{1/2}^{\otimes n}$ can be decomposed as
		\begin{align}
			\label{eq:tensordecomp}
			\mc Q_{1/2}^{\otimes n} \cong a_{n/2}\mc Q_{n/2}\oplus a_{n/2-1}\mc Q_{n/2-1}\oplus a_{n/2-2}\mc Q_{n/2-2}\oplus\dots,
		\end{align}
		where $\mc Q_J=\mathrm{Sym}^{2J}\mc Q_{1/2}$ is the spin-$J$ irreducible representation space of dimension $2J+1$ and we write $a_J\mc Q_J$ to mean $\bigoplus_{k=1}^{a_J}\mc Q_J$, i.e., each irreducible representation $\mc Q_J$ appears with multiplicity $a_J$.
		The Schur transform $U_{\rm Sch}$ changes the computational basis to the angular momentum basis.
		Explicitly, given an input state $\ket x\in\mc Q_{1/2}^{\otimes n}$ along with two ancilla registers initialized to $\ket{J=0}$, $\ket{q=0}$, the algorithm outputs a linear combination of Schur basis states with registers $\ket J, \ket q, \ket{p_1},\dots,\ket{p_n}$. 
		The register $\ket J$ corresponds to the total angular momentum $J$, labeling the spin-$J$ irreducible representation $\cQ_{J}$ of dimension $2J+1$. 
		The register $\ket q$ labels the basis vectors in the $\cQ_{J}$ representation.
		Since there are $a_J$ copies of each $\mc Q_J$ in $\mc Q_{1/2}^{\otimes n}$, the $\ket{p_i}$ registers label the different copies of $\mc Q_J$ in the decomposition. 
		It is crucial for implementing the evolution operator that in each copy, the $\ket q$ registers label the same vectors. For example, $\ket{q=J}$ is always the highest weight vector of the representation. 
		The circuit is shown in Fig.~\ref{fig:USch}.
		
		\begin{figure} [htbp]
			\centerline{\Qcircuit @C=1em @R=.7em {
					\lstick{\ket 0} & {/}\qw & \multigate{4}{U_{\rm Sch}} & \rstick{\ket J} \qw \\
					\lstick{\ket 0} & {/}\qw & \ghost{U_{\rm Sch}}& \rstick{\ket q} \qw \\
					\lstick{\ket{x_1}} & \qw & \ghost{U_{\rm Sch}} & \rstick{\ket{p_1}} \qw \\
					\lstick{\vdots} & \qw & \ghost{U_{\rm Sch}} & \rstick{\vdots} \qw \\
					\lstick{\ket{x_n}} & \qw & \ghost{U_{\rm Sch}} & \rstick{\ket{p_n}} \qw
			}}
			\caption{\label{fig:USch}The Schur transformation circuit. On an input $\ket x$, the output is a linear combination of the basis states $\ket J\otimes \ket q\otimes \ket p$.}
		\end{figure}
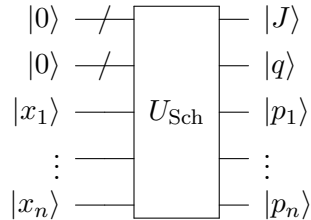
		
		The $\ket J$ and $\ket q$ registers are encoded in unary, where a basis state is $\ket 1$ on the qubit corresponding to the value of $J$ or $q$ and $\ket 0$ on all other qubits. For example, the values of $\ket J$ are encoded as $\ket{J=0} = \ket{100\dots}$, $\ket{J=\frac{1}{2}} = \ket{010\dots}$, $\ket{J=1} = \ket{001\dots}$ and the values of $\ket q$ are encoded as $\ket{q=-\frac{n}{2}} = \ket{100\dots}$, $\ket{q=-\frac{n-1}{2}} = \ket{010\dots}$, $\ket{q=-\frac{n-2}{2}} = \ket{001\dots}$. This way, a gate conditioned on $J=J_0$ or $q=q_0$ is controlled on one qubit instead of all the qubits representing $\ket J$ or $\ket q$. Since $n/2$ is the maximum value of $J$, we will need $\cO(n)$ qubits to represent $J$ and $q$ in this way, but the space $\cH$ already required $n$ qubits to represent, so the unary encoding does not significantly increase the space complexity.
		
		The Schur transformation circuit is composed of $n$ subroutines of Clebsch-Gordan transforms applied sequentially (Fig.~\ref{fig:USch detailed}) that are decomposed as in Fig.~\ref{fig:UCG detailed}. 
		In Fig.~\ref{fig:UCG detailed}, the $X$ gates on the $\ket J$ and $\ket q$ registers denote the transformation $\ket a\mapsto\ket{a+\frac 1 2}$. 
		In the unary encoding, it is implemented as a series of swap gates (Fig.~\ref{fig:X unary}). 
		Therefore, the $X^{-1}$ and controlled-$X^2$ gates on the $q$ and $J$ registers are composed of $\cO(n)$ elementary gates. 
		If we interpret the computational basis of the $\ket{x_i}$ register as $\ket 0=\ket{-\frac 1 2}$, $\ket 1=\ket{\frac 1 2}$, the $X^{-1}$ and controlled-$X^2$ gates together implement the controlled sum operation $\ket q\ket{x_i}\mapsto\ket{q+x_i}\ket{x_i}$ and $\ket J\ket{x_i}\mapsto\ket{J+x_i}\ket{x_i}$. 
		
		\begin{figure} [!htbp]
			\centerline{\Qcircuit @C=1em @R=.7em {
					\lstick{\ket 0} & {/}\qw & \multigate{2}{U_{CG}} & \multigate{1}{U_{CG}} & \multigate{1}{U_{CG}} & \multigate{1}{U_{CG}} & \rstick{\ket J} \qw \\
					\lstick{\ket 0} & {/}\qw & \ghost{U_{CG}} & \ghost{U_{CG}}\cwx[2] & \ghost{U_{CG}}\cwx[3] & \ghost{U_{CG}}\cwx[4] & \rstick{\ket q} \qw \\
					\lstick{\ket{x_1}} & \qw & \ghost{U_{CG}} & \ghost{} & \ghost{} & \ghost{} & \rstick{\ket{p_1}} \qw \\
					\lstick{\ket{x_2}} & \qw & \qw & \gate{{\color{white}U_{CG}}} & \ghost{} & \ghost{} & \rstick{\ket{p_2}} \qw \\
					\lstick{\vdots} & \qw & \qw & \qw & \gate{{\color{white}U_{CG}}} & \ghost{} & \rstick{\vdots} \qw \\
					\lstick{\ket{x_n}} & \qw & \qw & \qw & \qw & \gate{{\color{white}U_{CG}}} & \rstick{\ket{p_n}} \qw \\
			}}
			\caption{\label{fig:USch detailed}The Schur transformation $U_{\rm Sch}$ as a sequence of Clebsch-Gordan unitaries. The output is a linear combination of basis states $\ket J\otimes \ket q\otimes \ket p$.}
		\end{figure}
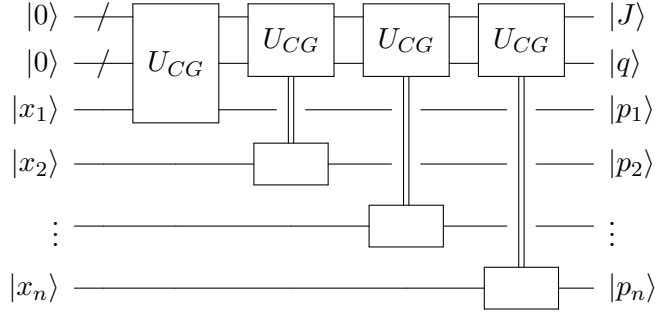
		
		\begin{figure} [!htbp]
			\centerline{\Qcircuit @C=1em @R=.7em {
					\lstick{\ket J} & {/}\qw & \qw & \qw & \ctrl{1} & \gate{X^{-1}} & \gate{X^2} & \rstick{\ket{J'}} \qw \\
					\lstick{\ket q} & {/}\qw & \gate{X^{-1}} & \gate{X^2} & \ctrl{1} & \qw & \qw & \rstick{\ket{q'}} \qw \\
					\lstick{\ket{x_i}} & \qw & \qw & \ctrl{-1} & \gate{R_y(\theta_{J,q'})} & \qw & \ctrl{-2} & \rstick{\ket{p_i}} \qw
			}}
			\caption{\label{fig:UCG detailed}A decomposition of the Clebsch-Gordan unitary $U_{CG}$ with output a linear combination of the basis states $\ket{J'}\otimes \ket{q'}\otimes \ket{p_i}$.}
		\end{figure}
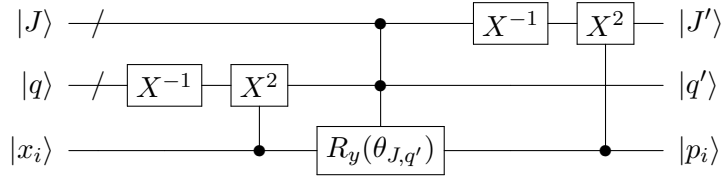
		
		\begin{figure} [!htbp]
			\centerline{\Qcircuit @C=1em @R=.7em {
					& \qw & \qw & \qw & \qswap & \qw\\
					& \qw & \qw & \qswap & \qswap \qwx & \qw\\
					\lstick{\ket q\text{ or }\ket J}& \qw & \qswap & \qswap \qwx & \qw & \qw\\
					& \qswap & \qswap \qwx & \qw & \qw & \qw\\
					& \qswap \qwx & \qw & \qw & \qw & \qw \gategroup{1}{1}{5}{1}{0.5em}{\{}
			}}
			\caption{\label{fig:X unary}An $X$ gate in the unary encoding used for the $q$ and $J$ registers.}
		\end{figure}
		
		The controlled $R_y$ gate of Fig.~\ref{fig:UCG detailed} denotes an operation that implements the rotation
		\begin{align}
			R_y(\theta_{J,q'}) = \begin{pmatrix} \cos\theta_{J,q'} & -\sin\theta_{J,q'}\\\sin\theta_{J,q'} & \cos\theta_{J,q'}\end{pmatrix}\;,
		\end{align}
		where
		\begin{align}
			\cos\theta_{J,q'} = \sqrt{\frac{1}{2}+\frac{q'}{2J+1}}\;, \quad \sin\theta_{J,q'} = \sqrt{\frac{1}{2}-\frac{q'}{2J+1}}\;.
		\end{align}
		It is decomposed as $\cO(n^2)$ rotation gates, each with only two controls because of the unary encoding. 
		An example is shown in Fig.~\ref{fig:CG rotation} with $n=3$ and $x_i=x_3$, the last instance of $U_{CG}$ in the decomposition of $U_{\rm Sch}$. 
		
		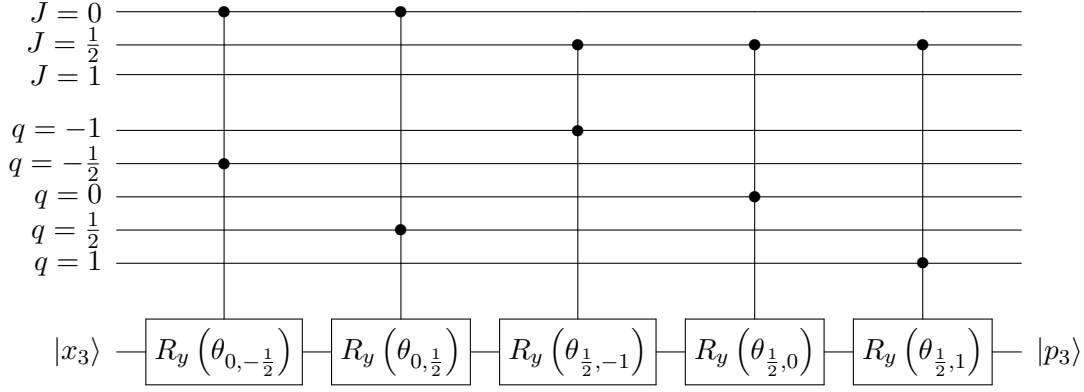
\begin{figure} [htbp]
			\centerline{\Qcircuit @C=1em @R=.9em {
					\lstick{J=0} & \ctrl{5} & \ctrl{7} & \qw & \qw & \qw & \qw \\
					\lstick{J=\frac{1}{2}} & \qw & \qw & \ctrl{3} & \ctrl{5} & \ctrl{7} & \qw \\
					\lstick{J=1} & \qw & \qw & \qw & \qw & \qw & \qw \\ \\
					\lstick{q=-1} & \qw & \qw & \ctrl{6} & \qw & \qw & \qw \\
					\lstick{q=-\frac 1 2} & \ctrl{5} & \qw & \qw & \qw & \qw & \qw \\
					\lstick{q=0} & \qw & \qw & \qw & \ctrl{4} & \qw & \qw \\
					\lstick{q=\frac 1 2} & \qw & \ctrl{3} & \qw & \qw & \qw & \qw \\
					\lstick{q=1} & \qw & \qw & \qw
					& \qw & \ctrl{2} & \qw \\
					\\
					\lstick{\ket{x_3}} & \gate{R_y\left(\theta_{0,-\frac 1 2}\right)} & \gate{R_y\left(\theta_{0,\frac 1 2}\right)} & \gate{R_y\left(\theta_{\frac 1 2,-1}\right)} & \gate{R_y\left(\theta_{\frac 1 2,0}\right)} & \gate{R_y\left(\theta_{\frac 1 2,1}\right)} & \rstick{\ket{p_3}} \qw
			}}
			\caption{\label{fig:CG rotation}The circuit for the rotation $R_y(\theta_{J,q'})$ controlled on the $J$ and $q$ registers. In the example illustrated, it implements the transformation $\ket J\otimes \ket{q'}\otimes \ket{x_3}\mapsto\ket J\otimes \ket{q'}\otimes \ket{p_3}$ when $n=3$. Note that we only need to perform the rotations for the valid values of $q'$ and $J$, satisfying $|q'|\le J+\frac{1}{2}$ and $q'\equiv J+\frac{1}{2} \pmod 1$.}
		\end{figure}
		
		We see from our decomposition that each $U_{CG}$ consists of $\cO(n^2)$ gates, giving a total of $\cO(n^3)$ gates to implement change of basis unitary $U_{\rm Sch}$.
	\end{proof}

	\section{Proof of Theorem~\ref{thm:FF&SEEM}}
	\label{appx:energy}
	We provide a proof of Thm.~\ref{thm:FF&SEEM} using the same method as in Ref.~\cite{AA17}, but we reproduce it here to calculate the quantum complexity for non-exponential fast-forwarding and precision of energy measurements.
	
	The main step in proving that fast-forwarding implies efficient energy measurements is the use of quantum phase estimation. This idea is captured in the following lemma.

	\begin{lemma}
		[Fast-forwarding $\to$ Precise energy measurements]
		\label{lem:FFtoSEEM}
		Suppose $H$ is a qubit Hamiltonian with $\left\|\left.H\right|_S\right\|\le 1$ for a subspace $\mc S\subseteq \left(\bb C^2\right)^{\otimes n}$, and $H$ can be simulated with parameters $(T,\epsilon,G)$ on $S$. For any $l,c\in \bb Z$ with $l\le \lfloor \log_2(2T)\rfloor$ and $3\le c<2^l$, $H$ has energy measurement parameters
		\begin{align}
			\left(\eta=1-\frac{1}{2(c-2)},\delta E=\frac{2\pi c}{2^{l}},\xi=l\epsilon,G'=\cO(l^2+lG)\right)
		\end{align}
		on $\mc S$.
	\end{lemma}
	\begin{proof}
		For simplicity, we give the proof for $\mc S=\left(\bb C^2\right)^{\otimes n}$, but the proof also holds for proper subspaces. First, consider the ideal case where we can implement $e^{-itH}$ (and therefore $e^{itH}$) exactly for $t\le T$. Since $\|H\|\le 1$, each eigenenergy $E$ satisfies $0\le E+1\le 2\le 2\pi$, so we can perform phase estimation on $e^{i(H+I)}$. To achieve $l$ bits of precision, we need to implementing controlled $e^{i(H+I)2^k}=e^{i2^k}e^{i2^kH}$ gates for $k=0,1,\dots,l-1$. Since $2^{l-1}\le T$ by assumption, each controlled $e^{i2^kH}$ can be implemented. The controlled $e^{i2^k}$ phase factors are phase gates on the control registers. The ideal circuit as described in Ref.~\cite{NC01} is shown in Fig.~\ref{fig:FF to SEEM}.
		
		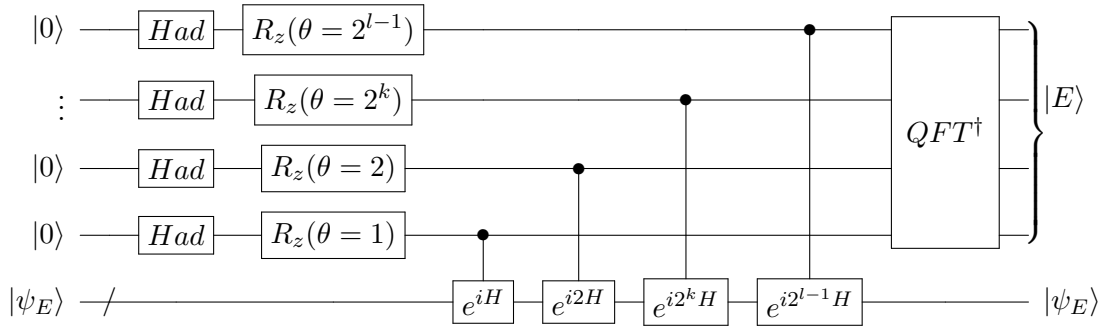
\begin{figure} [htbp]
			\centerline{\Qcircuit @C=1em @R=.7em {
					\lstick{\ket 0} & \qw & \gate{Had} & \gate{R_z(\theta=2^{l-1})} & \qw & \qw & \qw & \ctrl{4} & \multigate{3}{QFT^\dag} & \qw\\
					\lstick{\vdots} & \qw & \gate{Had} & \gate{R_z(\theta=2^k)} & \qw & \qw & \ctrl{3} & \qw & \ghost{QFT^\dag} & \rstick{\ket E}\qw\\
					\lstick{\ket 0} & \qw & \gate{Had} & \gate{R_z(\theta=2)} & \qw & \ctrl{2} & \qw & \qw & \ghost{QFT^\dag} & \qw\\
					\lstick{\ket 0} & \qw & \gate{Had} & \gate{R_z(\theta=1)} & \ctrl{1} & \qw & \qw & \qw & \ghost{QFT^\dag} & \qw\\
					\lstick{\ket{\psi_E}} & {/}\qw & \qw & \qw & \gate{e^{iH}} & \gate{e^{i2H}} & \gate{e^{i2^kH}} & \gate{e^{i2^{l-1}H}} & \qw & \rstick{\ket{\psi_E}}\qw \gategroup{1}{10}{4}{10}{0.5em}{\}}
			}}
			\caption{\label{fig:FF to SEEM}The quantum phase estimation circuit that uses Hamiltonian simulation subroutines $e^{i 2^k H}$. Each $Had$ is a Hadamard gate.}
		\end{figure}
		
		From Section 5.2.1 of Ref.~\cite{NC01}, phase estimation prepares a register $\ket m$ with
		\begin{align}
			\Pr(|m-\lfloor x\rfloor|>c)\le\frac{1}{2(c-1)}\;,
		\end{align}
		where $c$ is a positive integer and $x=2^l\left(\frac{E+1}{2\pi}\right)$ for an eigenstate of energy $E$. Thus,
		\begin{align}
			\Pr(|m-x|>c)&\le\Pr(|m-\lfloor x\rfloor|>c-1)\\
			& \le\frac{1}{2(c-2)}\;.
		\end{align}
		Rearranging the left hand side gives
		\begin{align}
			\Pr\left(\left|\frac{2\pi m}{2^l}-1-E\right|>\frac{2\pi c}{2^l} \right)\le\frac{1}{2(c-2)} \;, 
		\end{align}
		so that we have accuracy $\delta E=\frac{2\pi c}{2^l}$ and confidence bounded by $\eta = 1-\frac{1}{2(c-2)}$.
		
		In the (nonideal) circuit, there are $l$ Hadamard and $l$ single qubit rotation gates, $\cO(lG)$ elementary gates from the controlled time evolution operators $e^{i2^kH}$, and $\cO(l^2)$ elementary gates to implement the quantum Fourier transform. Therefore, the total gate count is $\cO(l^2+lG)$. The error in the circuit is due to the nonideal time evolution operators from the Hamiltonian simulation algorithm. Since each of the $l$ instances contribute an error of at most $\epsilon$, the total error in the energy measurement circuit is at most $\xi=l\epsilon$.
	\end{proof}
	
	The converse statement that efficient energy measurements imply the ability to fast-forward makes use of the following lemma.
	
	\begin{lemma}[Precise energy measurements $\to$ fast-forwarding]
		\label{lem:SEEMtoFF}
		Suppose a qubit Hamiltonian $H$ has energy measurement parameters $(\eta,\delta E,\xi,G)$ on a subspace $\mc S\subseteq \left(\bb C^2\right)^{\otimes n}$. Then for any $T$, the Hamiltonian can be simulated with parameters
		\begin{align}
			(T,\epsilon=\eta\delta E T+2(1-\eta+\xi),\cO(G))\;.
		\end{align}
		on $\mc S$.
	\end{lemma}
	\begin{proof}
		Let $U$ be the ideal energy measurement circuit and $\tilde U$ be the approximation. In the ideal case (Fig.~\ref{fig:SEEM to FF}), we apply $U$, then apply the rotations $e^{-itE}$ for $t\le T$ on the $\ket E$ register, and finally undo the energy measurement by applying $U^\dagger$.
		
		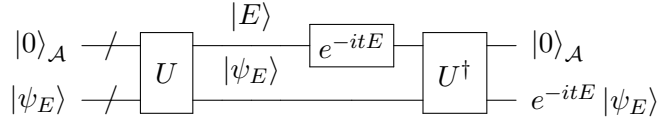
\begin{figure} [htbp]
			\centerline{\Qcircuit @C=1em @R=.7em {
					\lstick{\ket{0}_{\mc A}} & {/}\qw & \multigate{1}{U} & \qw & \ustick{\ket E}\qw & \qw & \gate{e^{-itE}} & \multigate{1}{U^\dag} & \rstick{\ket 0_{\mc A}} \qw\\
					\lstick{\ket{\psi_E}} & {/}\qw & \ghost{U} & \qw & \ustick{\ket{\psi_E}}\qw & \qw & \qw & \ghost{U^\dag} & \rstick{e^{-itE}\ket{\psi_E}} \qw 
			}}
			\caption{\label{fig:SEEM to FF}The quantum simulation circuit using energy measurements.}
		\end{figure}
		
		When we replace $U$ by $\tilde U$, there will be $2G$ elementary gates from the applications of $\tilde U$ and $\tilde U^\dagger$ and at most $G$ gates from the rotation $e^{-itE}$, since the number of ancillary qubits needed to store the energy is at most the number of gates in the algorithm. Note that the binary encoding of $\ket E$ allows us to implement $e^{-itE}$ in parallel as individual phase gates on the $\ket E$ register. Thus, the total gate count is $\cO(G)$.
		
		If the energy measurement had perfect confidence $\eta=1$ and zero error $\delta E=0$, the ideal circuit would exactly fast-forward the Hamiltonian. The error in the actual circuit comes from the imperfect values of $\eta$, $\delta E$, and $\xi$. Denote the rotations $e^{-itE}$ by the unitary $R$. For any eigenstate $\ket{\psi_E}$ of $H$ in the subspace $\cS$, we have 
		\begin{align}
			\left \|(e^{-itH}\otimes \one_\cA-U^\dagger RU)\ket{\psi_E}\otimes \ket 0_{\mc A}\right\|
			&= \left\|U^\dagger(e^{-itH}\otimes \one_\cA-R)U\ket{\psi_E}\otimes \ket 0_{\mc A}\right\|\\
			&= \|\ket{\psi_E}\otimes \sum_{E'}\alpha_{E'}(e^{-itE}-e^{-itE'})\ket{E'}_{\cA}\|\\
			&\le \eta\delta ET+2(1-\eta).
		\end{align}
		The last inequality is obtained by splitting the sum into two parts, $|E'-E|\le\delta E$ and $|E'-E|>\delta E$. For a general state $\ket\psi=\sum_{E}c_E\ket{\psi_E}\in \cS$, the same way of splitting the sum into two parts will give the same bound. Replacing $U$ by $\tilde U$ adds an additional $2\xi$ error. Therefore, the total error is $\alpha=\eta\delta ET+2(1-\eta+\xi)$.
	\end{proof}
	
	Combining the two lemmas, we prove Thm.~\ref{thm:FF&SEEM}.
	
	\begingroup
	\def\thetheorem{\ref{thm:FF&SEEM}}
	\begin{theorem}[Fast-forwarding and precise energy measurements]
		Let $\{H_n\}_n$ be a sequence of (qubit) Hamiltonians with norms $\|H_n\|\le 1$. Then the following statements hold:
		\begin{enumerate}
			\item If $\{H_n\}_n$ can be simulated with parameters $(T(n),\epsilon,G(n))$, then for any constant $\eta<1$, $\{H_n\}_n$ has energy measurement parameters $(\eta,\cO(\frac{1}{T(n)}),\cO(\epsilon\log T(n)),\cO((\log T(n))^2+G(n)\log T(n)))$.
			\item Conversely, if $\{H_n\}_n$ has energy measurement parameters $(\eta,\delta E(n),\xi,G(n))$, then for any constant $\alpha>0$, $\{H_n\}_n$ can be simulated with parameters $(\cO(\frac{1}{\delta E(n)}),\alpha\eta+2(1-\eta+\xi),\cO(G(n)))$.
		\end{enumerate}
		The same statements hold on subspaces $\mc S_n\subseteq \left(\bb C^2\right)^{\otimes n}$ if $\left\|\left.H_n\right|_{\mc S_n}\right\|\le 1$.
	\end{theorem}
	\addtocounter{theorem}{-1}
	\endgroup
	\begin{proof}
		To prove the first statement, use Lemma~\ref{lem:FFtoSEEM} with $l=\lfloor\log_2(2T(n))\rfloor$ and $c$ a large enough constant so that $1-\frac{1}{2(c-2)}>\eta$. This gives an accuracy $\delta E(n)$ that is inversely proportional to the fast-forwarding time $T(n)$ with gate count $\cO((\log T(n))^2+G(n)\log T(n))$.
		
		To prove the second statement, take $T(n)=\frac{\alpha}{\delta E(n)}$ in Lemma~\ref{lem:SEEMtoFF} to obtain fast-forwarding for time inversely proportional to the accuracy of the measurement with gate count $\cO(G(n))$.
	\end{proof}

\end{document}